\documentclass[journal,twocolumn]{IEEEtran}
\usepackage{times}
\usepackage{graphicx}
\usepackage{algorithm}
\usepackage{url}

\usepackage{cite}
\usepackage{amssymb}
\usepackage[cmex10]{amsmath}
\usepackage[tight,footnotesize]{subfigure}
\usepackage{amsmath}
\usepackage{bm}
\usepackage{color}
\usepackage{array}
\def\twon #1{\|#1\|}
\def\rainfty{\rightarrow\infty}

\def\argmin{\text{argmin}}

\def\bN{\mathbb{N}}

\def\bR{\mathbb{R}}

\def\cA{\mathcal{A}}
\def\cB{\mathcal{B}}

\def\cE{\mathcal{E}}

\def\cG{\mathcal{G}}

\def\cI{\mathcal{I}}

\def\cS{\mathcal{S}}
\def\cT{\mathcal{T}}
\def\cU{\mathcal{U}}

\def \qed {\hfill \vrule height6pt width 6pt depth 0pt}
\def\bee{\begin{equation}}
\def\ene{\end{equation}}
\def\beq{\begin{eqnarray}}
\def\enq{\end{eqnarray}}

\usepackage{etoolbox}
\let\bbordermatrix\bordermatrix
\patchcmd{\bbordermatrix}{8.75}{4.75}{}{}
\patchcmd{\bbordermatrix}{\left(}{\left[}{}{}
\patchcmd{\bbordermatrix}{\right)}{\right]}{}{}

\newtheorem{remark}{Remark}
\newtheorem{prop}{Proposition}

\newenvironment{proof}{\begin{IEEEproof}}{\end{IEEEproof}}

\begin{document}
\title{Range-based Coordinate Alignment for Cooperative Mobile Sensor Network Localization}
\author{Keyou You,~\IEEEmembership{Senior Member,~IEEE}, Qizhu Chen, Pei Xie,
and Shiji Song,~\IEEEmembership{Senior Member,~IEEE}
		
\thanks{This research was supported by  National Natural Science Foundation of China under Grants No.41576101 and No.41427806, and National Key Research and Development Program of China under Grant No.2016YFC0300801.}		
\thanks{K. You and S. Song are with the Department of Automation and BNRist,, Tsinghua University,  Beijing 100084, China. Email: \{youky, shijis\}@tsinghua.edu.cn.}
\thanks{ Q. Chen is with Beijing Science and Technology Co, three fast online, China. Email: chenqizhu@meituan.com.} 
\thanks{P. Xie is with the  JD.COM, China. Email:  xiepei13@jd.com.}}
\maketitle

\begin{abstract}
This paper studies a coordinate alignment problem for cooperative mobile sensor network localization with range-based measurements. The network consists of target nodes, each of which has only access position information in a {\em local} fixed coordinate frame, and anchor nodes with GPS position information. To localize target nodes, we aim to align their coordinate frames, which leads to a non-convex optimization problem over a rotation group $\text{SO}(3)$.  Then, we reformulate it as an optimization problem with a convex objective function over spherical surfaces. We explicitly design both iterative and recursive algorithms for localizing a target node with an anchor node, and extend to the case with multiple target nodes. Finally, the advantages of our algorithms against the literature are validated via simulations.   

\end{abstract}


%%%%%%%%%%%%%%%%%%%%%%%%%%%%%%%%%%%%\begin{comment}

\section {Introduction}\label{Introduction}
{{\em Cooperative}  localization is an important positioning technology \cite{patwari2005locating, wymeersch2009cooperative, kia2015cooperative,buehrer2018collaborative}.  
In the past decades, there are many methods for cooperative localization, such as semidefinite programming (SDP) \cite{wang2006further}, second-order cone programming \cite{tseng2007second}, sum of squares \cite{nie2009sum}, multidimensional scaling (MDS) \cite{shang2004localization}, convex relaxation \cite{soares2015simple} and parallel projection algorithms (PPA) \cite{gholami2013cooperative,jia2011set}. Among them, PPA is reported to yield comparable accuracy to SDP and MDS with much shorter running time \cite{jia2011set}, and is an attractive localization approach.}

By using both target-anchor  and target-target range measurements, this work is concerned with cooperative localization problems over mobile sensor networks where anchor nodes are encoded with GPS positions and each target node is only aware of its position information in a local fixed coordinate frame, whose orientation and position relative to the global frame of the GPS are unknown. This framework is of great importance in both the underwater \cite{chen2017cooperative} and aerial localization \cite{BOMIN17}.  
For example, in case of multiple autonomous underwater vehicles (AUVs) the GPS information is often available to a very limited number of AUVs. Then, it is sensible to use cooperative methods to localize other AUVs with the inter-AUV range measurements  \cite{bahr2009cooperative,papadopoulos2010cooperative,webster2012advances,wang2014optimization}. 
For unmanned aerial vehicles (UAVs), the target UAV in \cite{BOMIN17} is assumed to access to the Inertial Navigation System (INS), but the INS may continuously drift after initiation and lose the connection with the global coordinate system. That is, the GPS position of the target UAV is unavailable and requires to use inter-UAV range measurements for localization. 
 
 If a series of consistent positions in a local fixed coordinate frame can be obtained for a target node, its GPS position can be localized by aligning its local frame with the global frame of the GPS by using target-anchor measurements.  To this end, a natural way is to parameterize the local coordinate frame by a rotation matrix $R\in\text{SO}(3)$ and a translation vector $T\in\bR^3$. Then, the alignment problem reduces to the estimation of $(R, T)$, which is the key idea of \cite{chen2017cooperative,BOMIN17} and is also closely related to the idea of estimating the deviation of the local coordinate from the global coordinate in \cite{bahr2009cooperative,allotta2016development,huang2018new}.

{This work starts from investigating the problem of localizing a target node with an anchor node. The least squares estimate of $(R, T)$ can be obtained by solving an optimization problem with a non-convex objective function and non-convex constraints. Such a non-convex optimization problem in \cite{BOMIN17} is firstly relaxed as a SDP problem with $11$ equality constraints and the decision vector is a $17\times 17$ positive semi-definite matrix, hoping that the solution to the SDP problem can provide a good suboptimal solution.  To further refine the SDP solution, they design a gradient descent algorithm over the rotation group $\text{SO}(3)$. Differently from \cite{BOMIN17}, we exploit the geometric relations between nodes and reformulate the non-convex problem as a well-structured optimization problem with a convex cost over spherical constraints. This idea was presented in our conference paper \cite{chen2017cooperative}, the major results of which are all contained in Section III(A)-(B) of this work.}

{The striking feature of our approach is that we are able to simultaneously solve the coordinate alignment problem for multiple target nodes in a general sensor network by using both target-anchor  and target-target range measurements. Note that the authors in \cite{BOMIN17} only consider the case with only a target node, and is unclear how to extend to the general case with multiple target nodes. }

{With the aid of the block coordinate descent method \cite{bertsekas1999nonlinear}, we propose a parallel projection algorithm (PPA) to solve the above reformulated problem. The projection is with respect to the spherical surfaces and can be explicitly written in a simple form, after which the constraint $R\in\text{SO}(3)$ can also be easily resolved. Overall, the iteration of the PPA is given in a simple form and can be implemented with a low computational cost, which is important to the sensor network. In comparison with \cite{BOMIN17},  the PPA requires a much lower computational cost with comparable localization accuracy, both of which have been validated via numerical experiments.  }

{Interestingly, the PPA can easily incorporate new measurements  to update our estimate of $(R, T)$. Specifically,  we propose a recursive version of the PPA, which is termed as recursive projection algorithm (RPA), to approximately solve the optimization problem for coordinate alignment.  More importantly, we are able to extend our method to the case of multiple target nodes in a mobile sensor network. For a time-varying network, we further use the block coordinate descent idea to design the PPA to reduce the computational load. For a time-invariant network, we jointly use the Jacobi iterative method to run the PPA and obtain a distributed PPA, which only requires each target node to exchange information with its neighboring target nodes.  } 

%The main contributions of this paper are summarized as follows.
%\begin{itemize}
%	\item We provide a new formulation of simultaneously localizing multiple target nodes in a mobile sensor network by adopting the coordinate alignment.  
%	\item For the two-node localization problem, we propose a novel PPA to solve the coordinate alignment problem, whose advantages over the state-of-the-art works are validated via simulations. Moreover, we propose an online RPA, which seems impossible for the algorithm in \cite{BOMIN17}.
%	
%	\item A distributed PPA is proposed to localize multiple target nodes with a fixed communication topology, which is scalable to the network size and thus particularly useful for a large scale network.
%\end{itemize}

The rest of this paper is organized as follows. In Section \ref{Problem Statement}, we formulate the coordinate alignment over a time-varying network as a non-convex optimization. In Section \ref{twonode}, focusing on two-node coordinate alignment problem, we propose the PPA and RPA. In Section \ref{multinode}, we extend them to the multi-node setting and propose a PPA algorithm by using the block coordinate descent idea. For a fixed communication graph,  a distributed method with the Jacobi iteration is designed. The numerical experiments are conducted in Section \ref{Numerical Experiments}. Finally, some concluding remarks are drawn in Section \ref{Conclusion}.

\section{Problem Statement}\label{Problem Statement}
\subsection{The mobile sensor network}

The mobile sensor network is represented by a sequence of time-varying graphs $\mathcal{G}(t)=(\mathcal{V}, \mathcal{E}(t))$ where $\mathcal{V}$ is the set of a fixed number of mobile nodes and $\mathcal{E}(t)$ is the set of edges between nodes at discrete time $t\in\bN$. Specifically, $\mathcal{V}$ is the union of a target set $\mathcal{T}=\{1,\ldots,n\}$ and an anchor set $\mathcal{A}=\{n+1,\ldots,n+r\}$ where an anchor node can access its position information in the GPS  while a target node does not and is only aware of its position information in a local fixed coordinate frame whose orientation and position relative to the global frame of the GPS are unknown. See an example of  collaborative UAVs  in Section \ref{Introduction}. For brevity, the former is called the {\em global} position and the later is called {\em local} position.  Our objective is to localize the global positions of target nodes under information flow constraints, which are modeled by the graph $\mathcal{G}(t)$. 

Specifically, {for target node $i$ and anchor node $a$, $(i,a)\in\cE(t)$ identifies the communication from $a$ to $i$}. For any pair of target nodes $i$ and $j$ such that $(i,j)\in\cE(t)$, then $(j,i)\in\cE(t)$ and both nodes can communicate with each other. Moreover, two noisy range measurements $d_{ij}(t)$ and $d_{ji}(t)$ are taken by node $i$ and node $j$, respectively.  Note that $d_{ij}(t)$ and $d_{ji}(t)$ may not be equal due to the use of different range sensors. While for a target node $i$ and an anchor node $a$ such that $(i,a)\in\cE(t)$, only the range measurement $r_{ia}(t)$ is available to node $i$ and $(a,i)\notin \cE(t)$.   A target node $i$ is said to be connected to an anchor node $a$ in $\cG(t)$ if there is a path of consecutive edges in $\cE(t)$ that connects the two nodes. Given a target node $i$, let $\mathcal{T}_i(t)$ be the set of its neighboring target nodes, i.e., $\mathcal{T}_i(t) = \left\{j|j \in \mathcal{T}, (i,j)\in\cE(t)\right\}$ and $\mathcal{A}_i(t)$ is the set of neighboring anchor nodes, i.e., $\mathcal{A}_i(t) = \left\{a|a \in \mathcal{A}, (i,a)\in\cE(t)\right\}$.  Thus, the set of range measurements available to the target node $i$ at time $t$ is given as 
\bee\label{eqmeasure}
{m}_i(t)=\{d_{ij}(t),r_{ia}(t)|j\in\cT_i(t),a\in\cA_i(t)\}.
\ene

\subsection{Coordinate alignment for cooperative localization}
{Let $p_a^g(t) \in \mathbb{R}^3$ be the global position of an anchor node $a$ and $p_i^l(t) \in \mathbb{R}^3$ be the local position of a target node $i$ at time $t$, whose local coordinate system is parameterized by a rotation matrix $R_i^*\in \text{SO}(3)$ which is defined as 
$$\text{SO}(3)= \{R \in \mathbb{R}^{3 \times 3}| RR' = I, \text{det}(R) =1\}$$ and a translation vector $T_i^* \in \mathbb{R}^3$. Clearly,  the global position of the target node $i$ is expressed as
$
R_i^* p_i^l(t)+T_i^*. 
$
To localize the target node $i$, we aim to compute its coordinate parameters $(R_i^*,T_i^*)$ with noisy range measurements up to time $\bar{t}$, i.e.,
\begin{equation*}
\begin{aligned}
d_{ij}(t)&=\lVert R_i^*p_i^l(t) + T_i^* - R_j^*p_j^l(t) - T_j^* \rVert + \xi_{ij}(t),\\
r_{ia}(t)&=\lVert R_i^*p_i^l(t) + T_i^* -  p_a^g(t) \rVert + \xi_{ia}(t), t\in[1:\bar{t}],
\end{aligned}
\end{equation*}
where $e=(i,j)$ or $e=(i,a)$ indicates an edge and $\{\xi_{e}(t)\}_{t=1}^{\overline{t}}$ is a sequence of temporally uncorrelated with zero mean and the sequence $\{\xi_e(t)\}_{e\in\cE(t)}$ is spatially uncorrelated at any time $t$, and $[1:\bar{t}]=\{1,\ldots,\bar{t}\}$. Given a pair of  $R=[R_1',\ldots,R_n']'$ and $T=[T_1',\ldots,T_n']'$, the least squares estimate uses the quadratic loss function
\begin{equation}\label{eq:sub_objective_function}
\begin{split}
f_{ij}^{\mathcal{T}}(t,R,T) &= (d_{ij}(t)-\lVert R_ip_i^l(t) + T_i -  R_jp_j^l(t) - T_j \rVert)^2,\\
f_{ia}^{\mathcal{A}}(t,R,T) &= (r_{ia}(t)-\lVert R_ip_i^l(t) + T_i -  p_a^g(t) \rVert)^2.
\end{split}
\end{equation}}

Our coordinate alignment  problem for cooperative localization is formulated as a constrained optimization problem
\begin{equation}\label{eq:mle}
\begin{split}
&\underset{R_i, T_i}{\text{minimize}}~~ f(R,T):=\sum_{i=1}^n f_i(R, T)\\ 
& \text{\,subject to} \quad R_i \in \text{SO}(3) ,\ T_i \in \mathbb{R}^3, \forall i\in\cT
\end{split}
\end{equation}
where each summand in the objective function is given by
\begin{equation}\label{eq:objective_function}
f_i(R, T)= \sum_{t=1}^{\overline{t}}\big(
	                 \sum_{j\in\cT_i(t)}{f_{ij}^{\mathcal{T}}(t,R,T)}+ \sum_{a\in\cA_i(t)}{f_{ia}^{\mathcal{A}}(t,R,T)}\big).
\end{equation}

Under mild conditions, we show that the constrained optimization problem in \eqref{eq:mle} is solvable. 
\begin{prop} \label{prop_solvable}If each target node is connected to an anchor node in the union graph $\bigcup_{t=1}^{\overline{t}}\cG(t)$, then the constrained optimization problem in \eqref{eq:mle} contains at least an optimal solution.
\end{prop}
\begin{proof}See Appendix \ref{appendixa}. \end{proof}

 In the sequel, we shall design explicit algorithms to solve the optimization problem in \eqref{eq:mle} by using projection technique. 
\section{Localizing a Target Node with an Anchor Node}\label{twonode}

\begin{figure}
		\centering
		\includegraphics[width=2in]{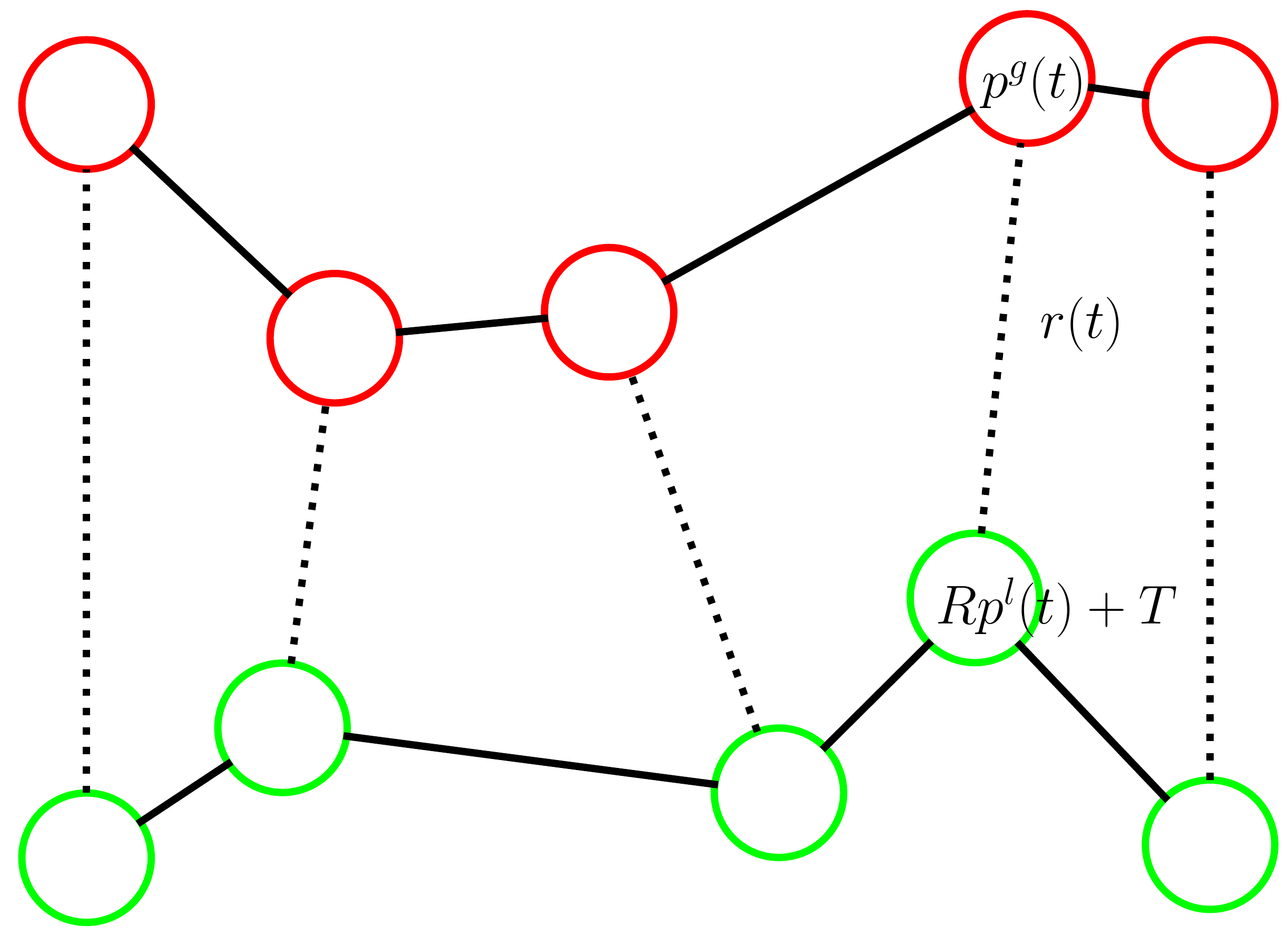}
	\caption{Two mobile AUVs with range measurements. One is a CNA and the other is a GPS-denied AUV. }
	\label{fig:two_moving_auv}
\end{figure}

In this section, we consider the problem of localizing only a mobile target node with an anchor node. This is well motivated by localizing a GPS-denied AUV. Another AUV with known global position is deployed to serve as a communication and navigation aid (CNA) \cite{huang2018new,bahr2009cooperative}.  They cooperatively work in the underwater and communicate with each other to obtain a series of range measurements, see Fig. \ref{fig:two_moving_auv}. {In this case, the minimum number of range measurements is $7$ \cite{yu2006principles}.}

To simplify notations of this section, let $p^g(t) \in \mathbb{R}^{3}$ be the global position of the anchor node, $p^l(t) \in \mathbb{R}^{3}$ be the local position of the GPS-denied target node and $r(t) \in \mathbb{R}$ be the range measurement between the two nodes at time $t$.  {Then, the information set for the target localization performed in the time interval $\bar{t}$} is collectively given by 
\bee\label{information2}
\cI(\overline{t})=\bigcup\nolimits_{t=1}^{\overline{t}} \{p^l(t),p^g(t),r(t)\},
\ene
and the optimization problem in (\ref{eq:mle}) is reduced as
\begin{eqnarray}\label{eq:mle1}
&&\underset{R,T}{\text{minimize}}~\sum\nolimits_{t=1}^{\overline{t}} f_t(R,T)  \nonumber \\ 
&& \text{subject to} \quad R \in \text{SO}(3), T\in\bR^3
\end{eqnarray}
where the summand in the objective function is
\begin{gather}\label{eq:t_th_summand}
f_t(R,T) = (r(t)-\lVert Rp^l(t) + T - p^g(t) \rVert)^2.
\end{gather} 
\subsection{Optimization problem reformulation using projection}
To solve the  optimization problem in (\ref{eq:mle1}), there are at least two challenges. The first is that $f_t(R,T)$ is non-convex, which usually is approximately solved by the convex relaxation  \cite{gholami2013cooperative,soares2015simple,BOMIN17,naseri2017cooperative}.  Here  we solve it by expressing as the minimization of a convex function over a spherical surface.  The second lies in the constraint set of a rotation group $\text{SO}(3)$, which fortunately can be explicitly solved as well. 

One can show that $f_t(R,T)$ is the squared range between the point $Rp^l(t) + T$ and the spherical surface centered at $p^g(t)$ with a radius $r(t)$\cite{soares2015simple}, see Fig. \ref{fig:projection}. That is,
\begin{equation}\label{eq:t_th_summand_inf_y}
\begin{aligned}
f_t(R,T) = \underset{y \in \cS(t)}{\min} {\lVert Rp^l(t) + T - y\rVert^2},
\end{aligned}
\end{equation}
where $\cS(t)$ is a spherical surface, i.e., 
\begin{equation}\label{eq:spherical_surfaces}
\cS(t)=\left\{y\in\bR^3|~\lVert y- p^g(t)\rVert=r(t)\right\}.
\end{equation}
\begin{figure}[t!]
		\centering
		\includegraphics[width=2in]{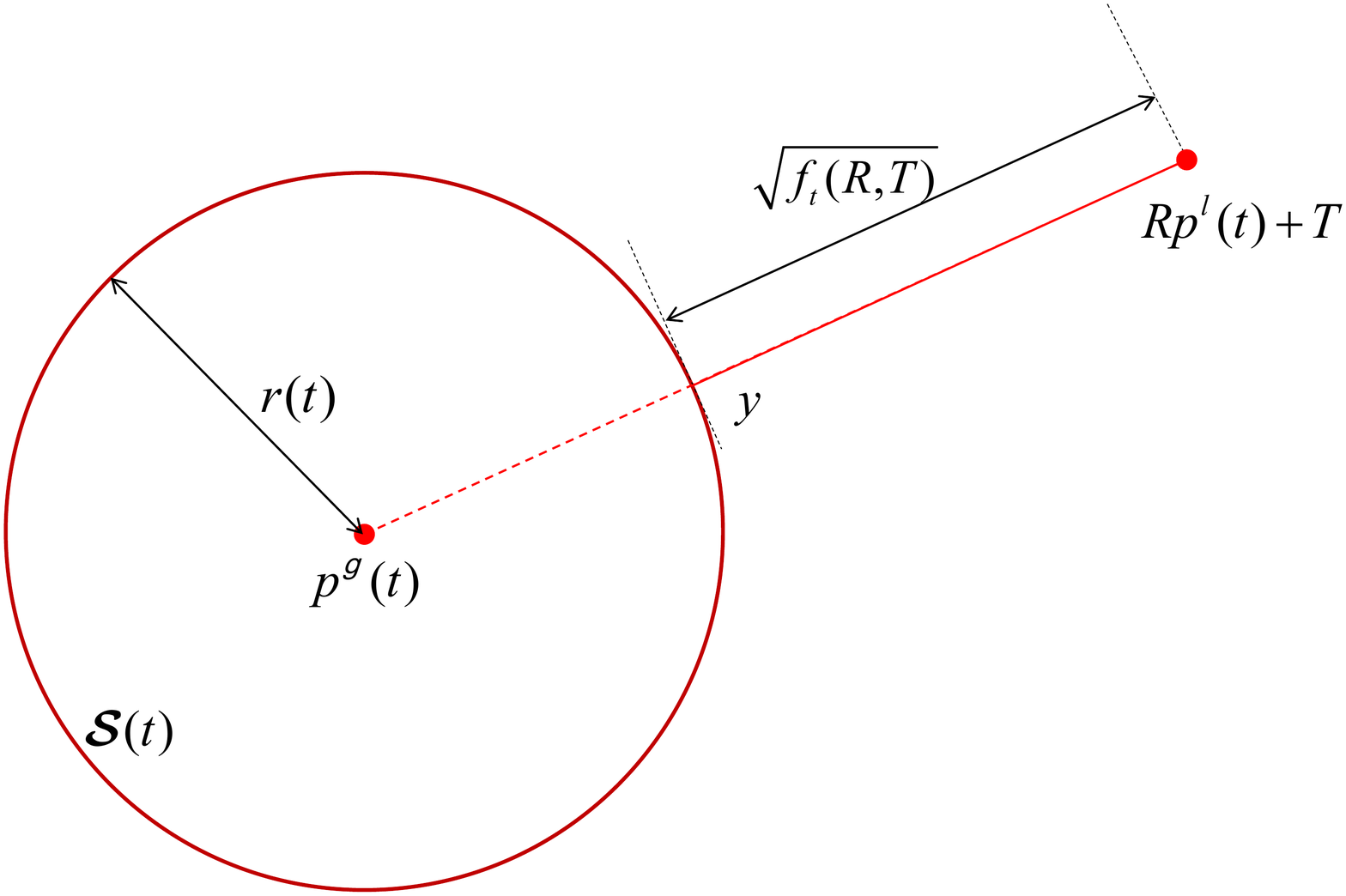}
	\caption{Projection onto a spherical surface if $r(t)<\lVert Rp^l(t) + T - p^g(t) \rVert$.}
	\label{fig:projection}
\end{figure}
In view of (\ref{eq:t_th_summand_inf_y}), we obtain the following optimization problem
\begin{equation}\label{eq:remle}
\begin{split}
& \underset{R,T,y_{1:\overline{t}}}{\text{minimize}} ~\sum\nolimits_{t=1}^{\overline{t}}{\lVert Rp^l(t) + T - y(t) \rVert^2} \\
& \text{subject to} \quad R \in \text{SO}(3), T \in \mathbb{R}^{3}, y(t) \in \cS(t), t\in[1:\overline{t}].
\end{split}
\end{equation}
\begin{remark}{Note that target localization is not instantaneous but performed in time interval $\bar{t}$.} In \cite{soares2015simple},  the so-called disk relaxation is adopted by relaxing the spherical surface $\cS(t)$ into a closed ball $\cB(t)=\{y|\lVert y- p^g(t)\rVert\le r(t)\}$. This leads to an underestimated convex problem, and is useless here as $\text{SO}(3)$ is not convex.
\end{remark}

Clearly, the two optimization problems in \eqref{eq:mle1} and \eqref{eq:remle} are essentially equivalent in the sense that both achieve the same minimum value and the same optimal set of $(R, T) $.  The good news is that the optimization problem \eqref{eq:remle} has favorable properties. First, its objective function is quadratically convex. Second, the newly introduced sets $\cS(t), t\in[1:\overline{t}]$ are spherical surfaces which are not difficult to compute the associated Euclidean projection. In fact, given a vector $y \in \mathbb{R}^3$, its Euclidean projection onto a spherical surface $\cS(t),t\in[1:\overline{t}]$ is explicitly expressed as
\begin{equation}\label{eq:projection_onto_ct}
P_{\cS(t)}(y) =p^g(t) + \frac{r(t)}{\lVert y-p^g(t) \rVert} (y-p^g(t))~\text{if}~y\notin \cS(t).
\end{equation}

To be specific, the projection of any matrix $\Omega \in \mathbb{R}^{3 \times 3}$ onto $\text{SO}(3)$ is obtained by solving a constrained optimization problem, i.e.,
\begin{equation*}
\begin{split}
 P_{\text{SO}(3)}(\Omega)& = \underset{R \in \text{SO}(3)}{\argmin} \lVert R-\Omega \rVert^2_F
 \end{split}
\end{equation*}
where $\twon{\cdot}_F$ denotes the Frobenius norm. In view of \cite{umeyama1991least},  $P_{\text{SO}(3)}(\Omega)$ is explicitly given as
\begin{equation}\label{eq:solution_of_projection_onto_so}
\begin{aligned}
P_{\text{SO}(3)}(\Omega)=UDV^{*},
\end{aligned}
\end{equation}
where $U$ and $V$ are obtained via the singular value decomposition of $\Omega$, i.e., $\Omega=U\Sigma V^*$, and 
\begin{equation*}
\begin{aligned}
D=
\begin{cases}
\text{diag}(1,1,+1),\text{ if } \text{det}(UV^{*})=1,\\
\text{diag}(1,1,-1),\text{ if } \text{det}(UV^{*})=-1.
\end{cases}
\end{aligned}
\end{equation*}

Next, we shall design algorithms to effectively solve the optimization problem (\ref{eq:remle}). 
\subsection{Parallel projection algorithm}
Once the target node has access the information set $\cI(\bar{t})$ in (\ref{information2}), it solves the optimization problem \eqref{eq:remle} by a block coordinate descent algorithm \cite{bertsekas1999nonlinear} with parallel projections. {We use master and worker to denote the order of updating per iteration}. Specifically, one master is used to update $(R^k,T^k)$ and $\overline{t}$-parallel workers are responsible for simultaneously updating $y^k(t), t\in[1:\overline{t}]$. The superscript $k$ denotes the number of iterations for solving the optimization problem \eqref{eq:remle}. 

At the $k$-th iteration, each worker receives the latest update $(R^k,T^k)$ from the master, and then performs the following projection in a parallel way
\begin{equation}\label{eq:y_parallel_projection}
\begin{aligned}
y^k(t) &= \argmin_{y(t)\in\cS(t)}\twon{R^kp^l(t)+T^k-y(t)} \\
&=P_{\cS(t)}(R^kp^l(t) + T^k),\quad t\in[1:\overline{t}],
\end{aligned}
\end{equation}
where $P_{\cS(t)}(\cdot)$ is given in (\ref{eq:projection_onto_ct}), and sends $y^k(t)$ to the master.  

Once the master receives $(y^{k}(1),\ldots,y^k(\overline{t}))$, it solves the following constrained least squares optimization 
\bee\label{master}
\underset{R\in\text{SO}(3),T}{\text{minimize}}\, \sum\nolimits_{t=1}^{\overline{t}}{\lVert Rp^l(t) + T - y^k(t) \rVert^2}.
\ene
\begin{prop}\label{prop_batch}
The optimization problem in (\ref{master}) is explicitly solved as 
\begin{equation}\label{eq:r_t_solution}
\begin{split}
R^{k+1} &= P_{\text{SO}(3)}(P^k),\\ 
T^{k+1} &= \overline{y}^k-R^{k+1}\overline{p}^l,
\end{split}
\end{equation}
where $P_{\text{SO}(3)}(\cdot)$ is given in (\ref{eq:solution_of_projection_onto_so}), $\overline{p}^l=\frac{1}{\overline{t}}\sum\nolimits_{t=1}^{\overline{t}}{p^l(t)}$ and $\overline{y}^k=\frac{1}{\overline{t}}\sum\nolimits_{t=1}^{\overline{t}}{y^k(t)}$ are ``mean" vectors of $\{p^l(t)\}$ and $\{y^{k}(t)\}$, and $P^k$ is their  ``correlation" matrix 
\bee\label{correlationmatri}
P^k=\sum\nolimits_{t=1}^{\overline{t}}{(y^k(t) - \overline{y}^k)(p^l(t) - \overline{p}^l)'}.
\ene
\end{prop}
\begin{proof}See Appendix \ref{appendixb}. \end{proof}
\begin{algorithm}[t!]  
	\caption{The Parallel Projection Algorithm (PPA) for Localizing a Target Node with an Anchor Node}
	\label{alg:ppa2} 
	\begin{enumerate}
	\renewcommand{\labelenumi}{\theenumi:}
	\item {\bf Input:} $\cI(\overline{t})$, which is the  information set for the target node, see (\ref{information2}).
\item {\bf Initialization:} The master arbitrarily selects $R^0\in\text{SO}(3)$ and $T^0\in\bR^3$, and sends to every worker $ t\in[1:\overline{t}]$.
\item {\bf Repeat}
\item {\bf Parallel projection:} Each local worker $t$ simultaneously computes 
$$y^k(t) =P_{\cS(t)}(R^kp^l(t) + T^k),\quad t\in[1:\overline{t}]$$
and sends $y^k(t)$ to the master. 

\item {\bf Master update:} The master computes the correlation matrix $P^k$ in (\ref{correlationmatri}) and uses (\ref{eq:solution_of_projection_onto_so}) to update as follows
 \begin{equation*}
\begin{split}
R^{k+1} &= P_{\text{SO}(3)}(P^k),\\ 
T^{k+1} &= \overline{y}^k-R^{k+1}\overline{p}^l.
\end{split}
\end{equation*}
\item {\bf Set} $k=k+1$.
\item {\bf Until} a predefined stopping rule (e.g., a maximum iteration number) is satisfied.
 \end{enumerate}
\end{algorithm}

 is interesting that \eqref{master} is closely related to the basic Procrustes problem \cite{Andersson97} and can be found in its full version in \cite{timmurphy}. For completeness, we also include a proof in Appendix. Finally, we summarize the above result in Algorithm \ref{alg:ppa2}.  
\begin{remark}Instead of using a SDP initialization \cite{BOMIN17}, we just randomly select a pair of $(R^0,T^0)$. Clearly, we can also adopt the same initialization to avoid getting into a bad local minimum.\qed \end{remark}

Since the optimization problem in (\ref{eq:remle}) is inherently non-convex, it cannot be guaranteed to converge to a global optimal solution.  However, it at least sequentially reduces the objective function per iteration, and achieves a better solution. To exposit  it, let $q:=(R,T,y_1,\ldots,y_{\bar{t}})$ and
$g(q):=\sum\nolimits_{t=1}^{\overline{t}}{\lVert Rp^l(t) + T - y(t) \rVert^2}$ be the decision variables and the objective function, respectively. We have the following result. 

\begin{prop} \label{prop_con}Let $\{q^k\}$ be iteratively computed in Algorithm \ref{alg:ppa2}. Then, it holds that 
$g(q^k)\le g(q^{k-1}), \forall k$ and there exists a convergent subsequence of $\{q^k\}$.
\end{prop}
\begin{proof}See Appendix \ref{appendixc}. \end{proof}
\begin{remark} In \cite{BOMIN17}, a semidefinite programming (SDP) relaxation is firstly devised to find an initial estimate of $(R,T)$, which involves solving a SDP with $11$ equality constraints and the decision vector is a $17\times 17$ positive semi-definite matrix. Then, they solve the optimization problem in (\ref{eq:mle1}) by using the projection of the gradient $M:=\frac{\partial}{\partial R}\sum\nolimits_{t=1}^{\overline{t}} f_t(R,T)  $ onto the tangent space of $\text{SO}(3)$, which is explicitly given as 
$$M_T(R)=\frac{1}{2}(M-RM'R).$$

The discretized version is essentially gradient descent (GD) and given by 
$R^{k+1}=P_{\text{SO}(3)}\left(R^k-\alpha^k M_T(R^k)\right),$ where $\alpha^k$ is a stepsize. Notably, they also explicitly state (without proof) that the SDP relaxation is important in providing a good initialization. Solving such a SDP and extracting a feasible $R^0\in \text{SO}(3)$ from the SDP's solution inevitably increases the computation cost.  {Though Algorithm \ref{alg:ppa2}  is only randomly initialized, numerical results show that its localization accuracy is still comparable to that of  \cite{BOMIN17}. }

More importantly, the focus of the equivalent optimization problem in (\ref{eq:remle}) allows us easily to devise a recursive algorithm to estimate $(R,T)$ in an online way (c.f. Section \ref{sec:recursive}) and generalize to the case of generic mobile sensor networks (c.f. Section \ref{multinode}). It is worthy mentioning that the approach in \cite{BOMIN17} currently only applies to a star topology.
\qed
\end{remark}

\subsection{The recursive projection algorithm}
\label{sec:recursive}
While Algorithm \ref{alg:ppa2} produces good results if $\overline{t}$ is moderately large, it does not exploit the sequential collection of the measurement, and the number of local intermediate variables $y(t)$ increases linearly with the number of range measurements. To resolve it,
 this subsection presents an {approximate} Recursive Projection Algorithm (RPA) which only performs one iteration whenever new measurement arrives.  

At time $t$, suppose we have already obtained a prior estimate $(R(t-1), T(t-1))$ and collected a new measurement  $\{p^l(t),p^g(t),r(t)\}$. Using this information, we shall recursively update the estimate of $(R,T)$ in an online  way. 

Similar to (\ref{eq:y_parallel_projection}), we perform an online projection 
\bee\label{projonline}
y(t)=P_{\cS(t)}(R(t-1)p^l(t) + T(t-1)),
\ene
where $\cS(t)$ is defined in (\ref{eq:spherical_surfaces}). In comparison with (\ref{master}),  the projection operation for $y(t)$ is only performed once.  Then,  the new estimate of $(R,T)$ is set as follows
\bee\label{online}
(R(t),T(t))=\underset{R\in\text{SO}(3),T}{\argmin}\, \sum\limits_{k=1}^{t}{\lVert Rp^l(k) + T - y(k) \rVert^2},
\ene
which can be recursively computed. 
\begin{prop}\label{prof_rpa}Let $\bar{y}(t), \bar{p}^l(t)$ and $P(t)$ be recursively computed by
\begin{equation}\label{ypl}
\begin{split}
\bar{y}(t)&= \bar{y}(t-1)+\frac{1}{t}(y(t)-\bar{y}(t-1)),\\ 
\bar{p}^l(t)&= \bar{p}^l(t-1)+\frac{1}{t}(p^l(t)-\bar{p}^l(t-1)),\\
P(t)&=P(t-1)+(1-\frac{1}{t})\\ 
&~~~~\times(y(t)-\bar{y}(t-1))(p^l(t)-\bar{p}^l(t-1))'
\end{split}
\end{equation}
where $\bar{y}(0)=\bar{p}^l(0)=0\in\bR^{3}$ and $P(0)=0\in\bR^{3\times 3}$. Then, the optimization problem in (\ref{online}) is solved by
\begin{equation}\label{onlinesolution}
\begin{split}
R(t)&= P_{\text{SO}(3)}(P(t)),\\ 
T(t) &=\bar{y}(t) -R(t) \bar{p}^l(t).
\end{split}
\end{equation}
\end{prop}
\begin{proof}See Appendix \ref{appendixd}. \end{proof}
\begin{algorithm}[t!]  
	\caption{The Recursive Projection Algorithm (RPA) for Localizing a Target Node with an Anchor Node}
	\label{alg:opa2} 
	\begin{enumerate}
	\renewcommand{\labelenumi}{\theenumi:}
\item {\bf Initialization:} The target node randomly selects  $R(0)\in\text{SO}(3)$ and $T(0)\in\bR^3$, and chooses $\bar{y}(0)=\bar{p}^l(0)=0\in\bR^{3},P(0)=0\in\bR^{3\times 3}$.
\item {\bf Online projection:} At time $t$, the target node receives a triple $\{p^l(t),p^g(t),r(t)\}$ and performs an online projection 
$$
y(t):=P_{\cS(t)}(R(t-1)p^l(t) + T(t-1)),
$$
where spherical surface $\cS(t)$ is defined in (\ref{eq:spherical_surfaces}).
\item {\bf Recursive update:} The target node recursively updates the triple $(\bar{y}(t),\bar{p}^l(t), P(t))$ by using (\ref{ypl}) and sets
\begin{equation*}
\begin{split}
R(t)&= P_{\text{SO}(3)}(P(t)),\\ 
T(t) &=\bar{y}(t) -R(t) \bar{p}^l(t).
\end{split}
\end{equation*}
\item {\bf Set} $t=t+1$.
 \end{enumerate}
\end{algorithm}

The recursive algorithm is summarized in Algorithm \ref{alg:opa2}. In practice, we shall further adopt the idea of smoothing \cite{anderson2012optimal} to improve the algorithmic performance. Instead of solving (\ref{online}), it is better to consider 
\bee\label{online1}
(R(t),T(t))=\underset{R\in\text{SO}(3),T}{\argmin}\, \sum\limits_{k=1}^{t}{\lVert Rp^l(k) + T - \acute{y}(k) \rVert^2},
\ene
where 
$
\acute{y}(k)=P_{\cS(k)}(R(t-1)p^l(k) + T(t-1))$ if $k\in[t-b, t-1]$ and $\acute{y}(k)=y(k)$ if $k<t-b$. Here $b$ denotes the length of smoothing interval and indicates the tradeoff between computational cost and performance improvement. Clearly,  Algorithm \ref{alg:opa2} corresponds to the special case $b=1$. Then, the optimization problem in (\ref{online1}) can be recursively solved. Let $\Delta(k)=\acute{y}(k)-y(k)$, which is zero if $k<t-b$, and compute $\bar{\Delta}(t)=\frac{1}{t}\sum_{k=1}^t \Delta(k) $. We solve it by replacing  $\bar{y}(t)$ and $P(t)$ with $\bar{y}(t)+\bar{\Delta}(t)$ and $$P(t)+\sum\nolimits_{k=t-b}^t(\Delta(k)-\bar{\Delta}(t))(p^l(k)-\bar{p}^l(t))'$$ in (\ref{onlinesolution}), respectively.

Since the localization problem is typically non-convex, we are unable to prove the asymptotic convergence of $(R(t),T(t))$. Jointly with (\ref{projonline}) and (\ref{online}), one may also use a discount factor $\alpha\in(0, 1)$ to emphasize the importance of the latest range measurements, e.g., 
$$
(R(t),T(t))=\underset{R\in\text{SO}(3),T}{\argmin}\, \sum\limits_{k=1}^{t}\alpha^{-k}{\lVert Rp^l(k) + T - y(k) \rVert^2}
$$
{and replace $1/t$ in \eqref{ypl} by $(1-\alpha)/(1-\alpha^t)$.}

\section{Localizing Multiple Target Nodes in the Sensor Network}\label{multinode}

In this section, we are interested in the localization problem of multiple target nodes in the mobile sensor network with generic time-varying communication topology $\cG(t)$. In \cite{BOMIN17}, the SDP based approach can only deal with the network setting that the only one anchor is connected to all target nodes. Such a scenario gives a star topology, which is trivial to treat by using the results on the situation with  one anchor and one target node. While for general mobile sensor networks, they leave it to future work. By using the approach in Section \ref{twonode}, we are able to  solve this problem, which is the focus of this section.
\subsection{Optimization problem reformulation using projection}
The loss function \eqref{eq:objective_function} introduces coupled summands, which makes the problem difficult. We shall use the projection idea in Section \ref{twonode} to reformulate the optimization problem in (\ref{eq:mle}). As in (\ref{eq:spherical_surfaces}), define the spherical surfaces
 \begin{equation*}
 \begin{split}
 \cS_{ij}(t)&=\{y\in\bR^3|\twon{y}=d_{ij}(t)\}, \forall  (i,j)\in\cE(t),\\
 \cS_{ia}(t)&=\{y\in\bR^3|\twon{y-p_a^g(t)}=r_{ia}(t)\},\forall (i,a)\in\cE(t).
 \end{split}
 \end{equation*} 
  In view of (\ref{eq:t_th_summand_inf_y}), the loss functions in (\ref{eq:sub_objective_function}) are rewritten as
\begin{equation}\label{eq:reform}
\begin{split}
&f_{ij}^{\mathcal{T}}(t) \\
&=  \min_{y_{ij}(t)\in\cS_{ij}(t)}\lVert R_ip_i^l(t) + T_i  - R_jp_j^l(t) - T_j - y_{ij}(t) \rVert^2,\\
&f_{ia}^{\mathcal{A}}(t) =\min_{y_{ia}(t)\in\cS_{ia}(t)} \lVert R_ip_i^l(t) + T_i  -  y_{ia}(t) \rVert^2.
\end{split}
\end{equation}
 
With a slight abuse of notations, let 
\begin{equation}\label{measupdate}
\begin{split}
y(t)&=\{y_{ij}(t),y_{ia}(t)\}_{(i,j)\in\cE(t),(i,a)\in\cE(t)},\\
\cS(t)&=\{\cS_{ij}(t),\cS_{ia}(t)\}_{(i,j)\in\cE(t),(i,a)\in\cE(t)}.
\end{split}
\end{equation}

Jointly with (\ref{eq:reform}), the problem in (\ref{eq:mle}) can be reformulated as
\begin{equation}\label{eq:multimle}
\begin{split}
&\underset{R, T, y(t), t\in[1:\overline{t}]}{\text{minimize}}~~  \sum_{t=1}^{\bar{t}}\sum_{i=1}^n  f_i(R,T,y(t)) \\ 
&~~ \text{subject to} \quad R \in\text{SO}(3)^n, y(t)\in\cS(t), t\in[1:\overline{t}],
\end{split}
\end{equation}
where the summand is given by
\begin{equation*}
\begin{split}
f_i(R,T,y(t))&=\sum_{a\in\cA_i(t)}\lVert R_ip_i^l(t) + T_i  -  y_{ia}(t) \rVert^2\\
&\hspace{-1cm}+\sum_{j\in\cT_i(t)}\lVert  R_ip_i^l(t) + T_i  - R_jp_j^l(t) - T_j  - y_{ij}(t) \rVert^2.
\end{split}
\end{equation*}

In the sequel, we shall design a block coordinate descent algorithm to solve the  optimization problem in (\ref{eq:multimle}). 
\subsection{Parallel projection algorithms}

Clearly, the objective function in (\ref{eq:multimle}) is quadratically convex.  We only need to handle the non-convex constraints $\text{SO}(3)$ and spherical surfaces $\cS(t), t\in[1:\overline{t}]$.

Now, we design a block coordinate descent algorithm \cite{bertsekas1999nonlinear} with parallel projections to solve \eqref{eq:multimle}. Specifically, given $(R^k, T^k)$, we update $y(t), t\in[1:\overline{t}]$ by 
\begin{equation*}
\begin{split}
&y^k(t)=\underset{y(t)\in\cS(t)}{\text{\argmin}}~~\sum_{i=1}^n  f_i(R^k,T^k,y(t)),
\end{split}
\end{equation*}
which can be explicitly expressed as
\begin{equation}\label{projectyij}
\begin{split}
y_{ij}^k(t)&=P_{\cS_{ij}(t)}( R_i^kp_i^l(t) + T_i ^k - R_j^kp_j^l(t) - T_j^k  ),\\
y_{ia}^k(t)&=P_{\cS_{ia}(t)}(R_i^kp_i^l(t) + T_i ^k), \forall (i,j), (i,a)\in\cE(t)
\end{split}
\end{equation}
and the projection onto a spherical surface is given in (\ref{eq:projection_onto_ct}). 

Next, we shall update $(R,T)$ by fixing $y(t)=y^k(t)$, i.e., 
\begin{equation}\label{updatert}
\begin{split}
\underset{R\in \text{SO}(3)^n, T}{\text{minimize}}~~ \sum_{i=1}^n \sum_{t=1}^{\bar{t}} f_i(R,T,y^k(t)),
\end{split}
\end{equation}

To solve the above optimization problem, the major difficulty lies in the constraints of $\text{SO}(3)$. Two ideas are adopted.
\subsubsection{Constrained least squares}
The first idea is to solve an unconstrained least squares  problem, i.e.,
\begin{equation}\label{unconsopt}
\begin{split}
(Z^k,T^{k+1})=\underset{Z, T}{\argmin}~~ \sum_{i=1}^n \sum_{t=1}^{\bar{t}} f_i(Z,T,y^k(t)),
\end{split}
\end{equation}
and then project $Z^k$ onto the constraints of $\text{SO}(3)^n$, i.e., 
\begin{equation*} 
R_i^{k+1}=\underset{ R_i\in\text{SO}(3)}{\argmin}~~  \twon{R_i-Z_i^k}_F^2,
\end{equation*}
which is explicitly given in (\ref{eq:solution_of_projection_onto_so}).

The remaining problem is how to effectively solve the least squares problem in (\ref{unconsopt}). For this purpose,  we represent $Z_ip_i^l(t) + T_i$ as a linear function of $\bm{x}_i$, where $\bm{x}_i \in \mathbb{R}^{12}$ is a column vector reshaping from $(Z_i,T_i)$. Specifically, {denote
$$B_i(t) = \begin{bmatrix} I_3\otimes p_i^l(t)', I_3 \end{bmatrix}, 
~ \bm{x}_i = \begin{bmatrix} \text{vec}(Z_i)', T_i' \end{bmatrix}',$$ where $\otimes$ denotes the Kronecker product, $\text{vec}(Z_i)\in\bR^{9}$ is a large vector by stacking all the columns of $Z_i \in\mathbb{R}^{3 \times 3}$,  and $I_3 \in \mathbb{R}^{3 \times 3}$ is an identity matrix. } Then, it follows that 
 $$Z_ip_i^l(t) + T_i = B_i(t)\bm{x}_i,$$
and the objective function in (\ref{unconsopt}) is rewritten as
\begin{equation*}
\begin{split}
 &f({\bm x})=\sum_{t=1}^{\bar{t}}\left(\sum_{(i,j)\in\cE(t)} \twon{B_i(t)\bm{x}_i -  B_j(t)\bm{x}_j-y_{ij}^k(t)}^2 +\right.\\
 &~~~~~~~~~~\left.\sum_{(i,a)\in\cE(t)}\twon{B_i(t)\bm{x}_i-y_{ia}^k(t)}^2\right),
\end{split}
\end{equation*}
which clearly is quadratic in the decision vector ${\bm x}$.

For a  graph $\cG(t)$, we define a  sparse block matrix $E(t)\in\bR^{|\cE(t)|\times n}\otimes \bR^{3\times 12}$ over the graph for a compact form of $f({\bm x})$. Particularly, if $e=(i,j)\in\cE(t)$ and $j\in\cT_i(t)$, then the $(e,i)$-th block of $E(t)$ is $B_i(t)$ and the $(e,j)$-th block of $E(t)$ is $-B_j(t)$. If $e=(i,a)\in\cE(t)$ and $a\in\cA_i(t)$, then the $(e,i)$-th block of $E(t)$ is $B_i(t)$. All the unspecified blocks are set to be zero matrices with compatible dimensions.  This implies that the objective function in (\ref{unconsopt})  can be compactly expressed as
\begin{equation*}
f({\bm x})= \sum\nolimits_{t=1}^{\overline{t}}{\lVert E(t)\bm{x} - {y}^k(t) \rVert^2}.
\end{equation*}
Clearly, the minimizer of $f({\bm x})$ is simply given by
\bee\label{ls}
{\bm x}^{ls}=\left(\sum_{t=1}^{\bar{t}}E(t)'E(t)\right)^{-1}\left(\sum_{t=1}^{\bar{t}}E(t)'{y}^k(t)\right) \in\bR^{12n}.
\ene

To compute the above ${\bm x}^{ls}$, let $Q(t)=E(t)'E(t) \in \mathbb{R}^{n\times n}\otimes \bR^{12\times 12}$. Denote the $(i,j)$-th block of $Q(t)$ by $Q(t)_{(i,j)}\in\bR^{12\times 12}$, it follows that
\begin{equation}\label{eq:qii}
Q(t)_{(i,i)}=(2|\cT_i(t)|+|\cA_i(t)|)(B_i'(t)B_i(t))
\end{equation}
and
\begin{equation}
Q(t)_{(i,j)}=
\left\{\begin{array}{ll}
-2B_i'(t)B_j(t),&\text{if}~ (i, j) \in \mathcal{E}(t),\\
0,&\text{otherwise.}
\end{array}\right.
\end{equation}
where $|\cT_i(t)|$ and $|\cA_i(t)|$ denote the cardinality of the sets $\cT_i(t)$ and $\cA_i(t)$ respectively, and
$$B_i'(t)B_j(t) = \begin{bmatrix} I_3 \otimes p^l_{i}(t)p^l_{j}(t)' & I_3 \otimes p^l_{i}(t)\\ I_3 \otimes p^l_{j}(t)' &I_3 \otimes I_3 \end{bmatrix}.$$

Similarly, the $i$-th block of $E(t)'{y}^k(t)$ is defined as $E(t)'{y}^k(t)_{(i)}$ and given by
\begin{equation*}
\begin{split}
B_i(t)\left(
\sum\limits_{j \in \cT_i(t)}
{
	(y^k_{ij}(t) - y^k_{ji}(t))
}+
\sum\limits_{a \in \cA_i(t)}
{
	y^k_{ia}(t)
}
\right)\in\bR^{12}.
\end{split}
\end{equation*}

%This further implies that
%\begin{equation*}
%\begin{split}
%E(t)' y^k(t)_{(i)}=B_i(t) \big( 2 \sum\limits_{j \in \cT_i(t)}y^k_{ij}(t)+\sum\limits_{a \in \cA_i(t)}y^k_{ia}(t)\big).
%\end{split}
%\end{equation*}
 Let
$
\widehat{y}_i^k(t)=\sum\nolimits_{j \in \cT_i(t)}\left(y^k_{ij}(t)-y^k_{ji}(t)\right)+\sum\nolimits_{a \in \cA_i(t)}{y^k_{ia}(t)}$,
then 
\begin{equation}\label{eq:ey_final}
E(t)' y^k(t)_{(i)}=\begin{bmatrix} \widehat{y}_i^k(t) \otimes p^l_{i}(t)\\ \widehat{y}_i^k(t) \end{bmatrix}.
\end{equation}

Jointly with \eqref{eq:qii}-\eqref{eq:ey_final}, the minimizer in \eqref{ls} can be readily computed.  If the graph $\cG(t)$ is fixed,   (\ref{ls}) can be cast as a sparse least squares problem, see e.g. \cite{fong2011lsmr} for details.

\begin{algorithm} [t!]
	\caption{The Parallel Projection Algorithm for Localizing Multiple Target Nodes}
	\label{alg:ppa3} 
	\begin{enumerate}
	\renewcommand{\labelenumi}{\theenumi:}
	\item {\bf Input:} Every target node $i$ collects the information $\cI_i(\overline{t})=\bigcup_{t=1}^{\bar{t}}\{m_i(t),p_i^l(t),p_a^g(t)|i\in\cT_i(t),a\in\cA_i(t)\}$, where $m_i(t)$ is defined in (\ref{eqmeasure}). A master (fusion center) collects time-varying graphs $\bigcup_{t=1}^{\bar{t}}\cG(t)$.
\item {\bf Initialization:} The master arbitrarily selects $R_i^0\in\text{SO}(3)$ and $T_i^0\in\bR^3$, and sends to each target node $i\in\cT$.
\item {\bf Repeat}
\item {\bf Parallel projection:} Each target node $i$ simultaneously computes 
$$\{y_{ij}^{k}(t), y_{ia}^k(t)|j\in\cT_i(t),\cA_i(t)\},\quad t\in[1:\overline{t}]$$
and $E(t)' y^k(t)_{(i)}$ by using (\ref{projectyij}) and (\ref{eq:ey_final}), respectively, and send $E(t)' y^k(t)_{(i)}$ to the master. 

\item {\bf Master update:} The master computes the least squares vector in (\ref{ls}), which is the solution of (\ref{unconsopt})  and obtains $(Z^k,T^{k+1})$. Then, it sets
$$
R_i^{k+1}=P_{ \text{SO}(3)}(Z_i^k)
$$
by using (\ref{eq:solution_of_projection_onto_so}) and sends $(R_i^{k+1},T_i^{k+1})$ to each target node $i\in\cT$.
\item {\bf Set} $k=k+1$.
\item {\bf Until} a predefined stopping rule (e.g., a maximum iteration number) is satisfied.
 \end{enumerate}
\end{algorithm}

\subsubsection{Jacobi iterative method}
We can also solve the optimization problem in (\ref{updatert}) by using the Jacobi iterative method \cite{bertsekas1999nonlinear}.  Particularly,  we compute $(R_i, T_i)^{k+1}$  by setting $(R_{-i},T_{-i})$ to be $(R_{-i},T_{-i})^k$, where $(R_{-i},T_{-i})=\{(R_j,T_j)\}_{j\in\cT, j\not= i}$, i.e., 
\begin{equation}\label{updatert1}
\begin{split}
(R_i,T_i)^{k+1}&=\underset{R_i\in \text{SO}(3), T_i}{\argmin}~~\sum_{t=1}^{\bar{t}}g_i(R,T,y^k(t))\\
&~~~\text{subject to}~(R_{-i},T_{-i})=(R_{-i},T_{-i})^k,
\end{split}
\end{equation}
where the objective  collects all summands in the objective function of (\ref{updatert}) containing the decision variables $(R_i, T_i)$, and $g_i(R,T,y^k(t))$ is given by
\begin{equation*}\begin{split}
g_i(R,&T,y^k(t))=\sum_{a\in\cA_i(t)}\twon{R_ip_i^l(t)+T_i-y_{ia}^k(t)}^2\\
&+\sum_{j\in\cT_i(t)}\twon{R_ip_i^l(t)+T_i-R_j^kp_j^l(t)-T_j^k-y_{ij}^k(t)}^2.
\end{split}
\end{equation*}

Then, the optimization problem in \eqref{updatert1}  has a similar structure to that of \eqref{master}, and can be solved as
\begin{equation}\label{eq:r_t_i_solution}
\begin{split}
R_i^{k+1} &= P_{\text{SO}(3)}(P_i^k),\\ 
T_i^{k+1} &= \overline{y}_i^k-R_i^{k+1}\overline{p}_i^l,
\end{split}
\end{equation}
where $\overline{y}_i^k$ and $\overline{p}_i^l$ are two mean vectors and $P_i^k$ is a correlation matrix, i.e., 
\begin{equation*}
\begin{split}
\overline{y}_i^k &= \frac{1}{n_i}\sum\nolimits_{t=1}^{\overline{t}}{\left(\sum\nolimits_{j \in \cT_i(t)}{\breve{y}^k_{ij}(t)}+\sum\nolimits_{a \in \cA_i(t)}{y^k_{ia}(t)}\right)},\\ 
\overline{p}_i^l &= \frac{1}{n_i}\sum\limits_{t=1}^{\overline{t}}{\left(|\cT_i(t)|+|\cA_i(t)|\right)p_i^l(t)},\\ 
P_i^k &= \sum\nolimits_{t=1}^{\overline{t}}{\left(\sum\nolimits_{j \in \cT_i(t)}{P^k_{ij}(t)}+\sum\nolimits_{a \in \cA_i(t)}{P^k_{im}(t)}\right)},\\
n_i&= \sum\nolimits_{t=1}^{\overline{t}}{\left(|\cT_i(t)|+|\cA_i(t)|\right)},\\ 
\breve{y}^k_{ij}(t) &= y^k_{ij}(t)+R_j^kp_j^l(t)+T_j^k,\\
P^k_{ij}(t) &= (\breve{y}^k_{ij}(t) - \overline{y}_i^k)(p_i^l(t) - \overline{p}_i^l)',\\
P^k_{ij}(t) &= (y^k_{ia}(t) - \overline{y}_i^k)(p_i^l(t) - \overline{p}_i^l)'.
\end{split}
\end{equation*}

Different from (\ref{ls}), the Jacobi iterative method does not need to solve the  least squares problem in (\ref{unconsopt}), which may need to compute the inverse of a large matrix, i.e., $\sum_{t=1}^{\bar{t}}E(t)'E(t)\in \mathbb{R}^{n\times n}\otimes \bR^{12\times 12}$. Instead,  we only need to use (\ref{eq:r_t_i_solution}) to replace Step 5 in Algorithm \ref{alg:ppa3}. 

\subsection{Distributed implementation of Jacobi method for fixed graphs}
Centralized algorithms are not scalable for the large network.  If $\cG(t)$ is fixed,  the Jacobi iterative method can even be implemented in a distributed way, which is termed as DPPA and given in Algorithm \ref{alg:ppa4}. 
{\begin{remark} By Proposition 2.3.1 \cite{bertsekas1999nonlinear}, we can obtain the similar result as Proposition \ref{prop_con} for Algorithms \ref{alg:ppa3}-\ref{alg:ppa4}. Take Algorithm \ref{alg:ppa3} as an example. Let  ${\bf y}=[y(1),\ldots,y(\bar{t})]$ and $$g(R, T,{\bf y})=\sum_{i=1}^n \sum_{t=1}^{\bar{t}} f_i(Z,T,y(t)).$$ Then, there is a convergent subsequence of $\{(R^k, T^k, {\bf y}^k)\}$. To elaborate it, we obtain from \eqref{unconsopt} that $g(Z^k,T^{k+1},{\bf y}^k)\le g(R^k,T^{k},{\bf y}^k)$. Since the projection operator is non-expansive, it implies that $g(R^{k+1},T^{k+1},{\bf y}^k)\le g(Z^k,T^{k+1},{\bf y}^k)$.  By \eqref{projectyij}, it holds that $g(R^{k+1}, T^{k+1},{\bf y}^{k+1})\le g(R^{k+1}, T^{k+1},{\bf y}^{k})$. Combining the above, it finally yields that 
$$g(R^{k+1}, T^{k+1},{\bf y}^{k+1})\le g(R^k,T^{k},{\bf y}^k).$$
The rest of proof follows exactly the same as that of Proposition \ref{prop_con}. \qed
\end{remark}}

\begin{algorithm} [t!]
	\caption{The Distributed PPA (DPPA) for Localizing Multiple Target Nodes in a Fixed Graph}
	\label{alg:ppa4} 
	\begin{enumerate}
	\renewcommand{\labelenumi}{\theenumi:}
	\item {\bf Input:} Every target node $i$ collects the information $\cI_i(\overline{t})=\bigcup_{t=1}^{\bar{t}}\{m_i(t),p_i^l(t),p_a^g(t)|i\in\cT_i,a\in\cA_i\}$, where $m_i(t)$ is defined in (\ref{eqmeasure}). 
	\item {\bf Initialization:} Every target node $i$ randomly selects $R_i^0\in\text{SO}(3)$ and $T_i^0\in\bR^3$, and then broadcasts to its neighboring target nodes $j\in\cT_i$.
	\item {\bf Repeat}
	\item {\bf Distributed update:} Each target node $i$ simultaneously computes 
$$\{y_{ij}^{k}(t), y_{ia}^k(t)|j\in\cT_i,\cA_i\},\quad t\in[1:\overline{t}]$$
by using (\ref{projectyij}), and $(R_i^{k+1},T_i^{k+1})$ by using $(\ref{eq:r_t_i_solution})$. Then, it broadcasts $(R_i^{k+1},T_i^{k+1})$ to its neighboring target nodes $j\in\cT_i$.
\item {\bf Set} $k=k+1$.
\item {\bf Until} a predefined stopping rule (e.g., a maximum iteration number) is satisfied.
 \end{enumerate}
\end{algorithm}

\section{Numerical Experiments}\label{Numerical Experiments}
In this section, we perform numerical experiments to validate the proposed algorithms in Python 2.7 environment on a MacBook Pro with 2.2 GHz Intel Core i7 CPU and 16GB DDR3. Open source packages such as Numpy 1.12.1 and cvxopt 1.1.9 are used for numerical computation. The experiments are implemented in both two dimensional space and three dimensional space. As there is no difference between the two cases, we only report results of the two dimensional case for visualization convenience.

\subsection{Experiment setup}
 For the two-node localization problem, the coordinate system of the target node is generated by a rotation matrix $R$ and a transformation vector $T$ as follow
$$
R = \left[\begin{matrix}\cos\theta&-\sin\theta\\ \sin\theta&\cos\theta\end{matrix}\right],~T = \left[\begin{matrix}a\\b\end{matrix}\right],
$$
where the rotation angle $\theta \sim \cU(0, 2\pi)$, and $a,b\sim \cU(-1,1)$ are randomly selected with uniform distributions.  The target node and the anchor node are randomly moving in a  square area $[1,9]\times[1,9]$. We also randomly select $p^l(t)$ and $p^g(t)$ such that $p^g(t)\in[1,9]\times[1,9]$ and $Rp^l(t)+T\in[1,9]\times[1,9]$. 

Then their range measurements at time slot $t$ are generated by $r(t) = \|Rp^l(t)+T-p^g(t)\|+\xi$ where the random noise is $\xi\sim N(0,\sigma^2)$. To quantify the noise level, define the signal-to-noise ratio (SNR) by
$$
\text{SNR}_{dB}=10\log_{10}\left(\frac{d_0^2}{\sigma^2}\right),
$$
where $d_0=4.1712$ is the average range of two nodes in the area $[1,9]\times[1,9]$. Clearly, a smaller SNR means a higher noise level. Our objective is to compute the coordinate system parameters under different signal-to-noise ratios by the proposed algorithms, which are denoted as $\widehat{R}$ and $\widehat{T}$. We are concerned with their relative errors 
\bee\label{rela-error}
err_R=\frac{\|\widehat{R}-R\|_F}{\|R\|_F},~err_T=\frac{\|\widehat{T}-T\|}{\|T\|}.
\ene

For each target node $i\in\mathcal{T}$ in the multi-node localization problem, $\widehat{R}_i$ and $\widehat{T}_i$ are denoted as the same way as that in the two-node localization scenario. Similarly, all nodes are limited to the square area $[1,9]\times[1,9]$.

\subsection{Experimental results of the two-node localization problem}
\label{Two-Agent Localization Simulations}
%Let $\text{SNR}=10$ and we run the PPA in Algorithm \ref{alg:ppa2} up to $60$ iterations and the result is shown  in Fig.\ref{fig:convergence}. Clearly, both the rotation angle $\theta$ and the translation vector $T$ are found. 
%\begin{figure}[ht]
%	\centering
%	\subfigure[Trajectory of $T$]{ % https://tug.org/TUGboat/tb34-1/tb106thurnherr.pdf
%		\includegraphics*[width=1.6in]{ConvergenceT_Pred_Trajectory1.pdf}
%		\label{fig:t_pred_trajectory}}
%	%\quad
%	\subfigure[Rotation angle of $R$]{
%		\includegraphics*[width=1.6in]{Convergence_Theta1.pdf}
%		\label{fig:theta}}
%	\caption{(a) The left figure illustrates the convergence of $T$ in the two dimensional space to the black circle, which represents the true translation vector $T$. (b) The right figure illustrates the convergence of the rotation angle to the black line, which represents the rotation angle of $R$.}
%	\label{fig:convergence}
%\end{figure}

We compare the proposed PPA with the SDP based method \cite{BOMIN17}. Since the method in \cite{BOMIN17} is unable to deal with general multi-node situations, we only compare their algorithm for localizing one target node with one anchor node.

Numerical experiments are performed under two noise levels\footnote{Kindly note that the results of $\text{SNR}=10, 20, 30$ are consistent, we only report the case of $\text{SNR}=20$ for saving space.} ($\text{SNR}=20$ and $\text{SNR}=80$) and three different $\overline{t}$. The results  in Fig. \ref{fig:compare_ppa_sdp} are obtained by averaging over $10^4$ independent simulations. Since rotation matrices are more difficult to estimate, we choose to report results mostly on rotations and only include the final results on translations for saving space.  The green line corresponds to the use of the pure SDP, and the red line is the result of the PPA of Algorithm \ref{alg:ppa2}. The blue line is the result of  the GD with the SDP initialization \cite{BOMIN17}, i.e., SDP+GD, while the purple line is the result of the PPA with the SDP initialization, i.e., SDP+PPA.  We also record the time used for running different algorithms. In Fig. \ref{fig:compare_ppa_sdp}(a), it takes 4.96e-01s to find the SDP based solution. To achieve the same relative rotation error, it only takes 2.19e-03s by using PPA. Moreover, it only takes 4.20e-03s for PPA to outperform the SDP+GD, whose running time is (4.96e-01+6.88e-03)s. We also observe that the SDP+PPA finally achieves the smallest relative rotation error. If the SNR is large, see Fig.\ref{fig:compare_ppa_sdp}(b), the PPA cannot reduce the relative rotation error as small as that of the SDP {due to the use of random initialization and  the gap induced by the SDP relaxation decreases with SNR.} However, the SDP+PPA performs much better than the SDP+GD, both in terms of running time and accuracy. In Table \ref{ret}, we include the final results, i.e., the number of iterations is set to $1000$, on the relative translation errors when $\bar{t}=20$.   In summary, both Fig. \ref{fig:compare_ppa_sdp} and  Table \ref{ret} consistently validate the advantages of the PPA of Algorithm \ref{alg:ppa2}.

\begin{figure}[t!]
	\centering
	\subfigure[$\overline{t}=10$, $\text{SNR}=20$]{ % https://tug.org/TUGboat/tb34-1/tb106thurnherr.pdf
		\includegraphics*[width=1.6in]{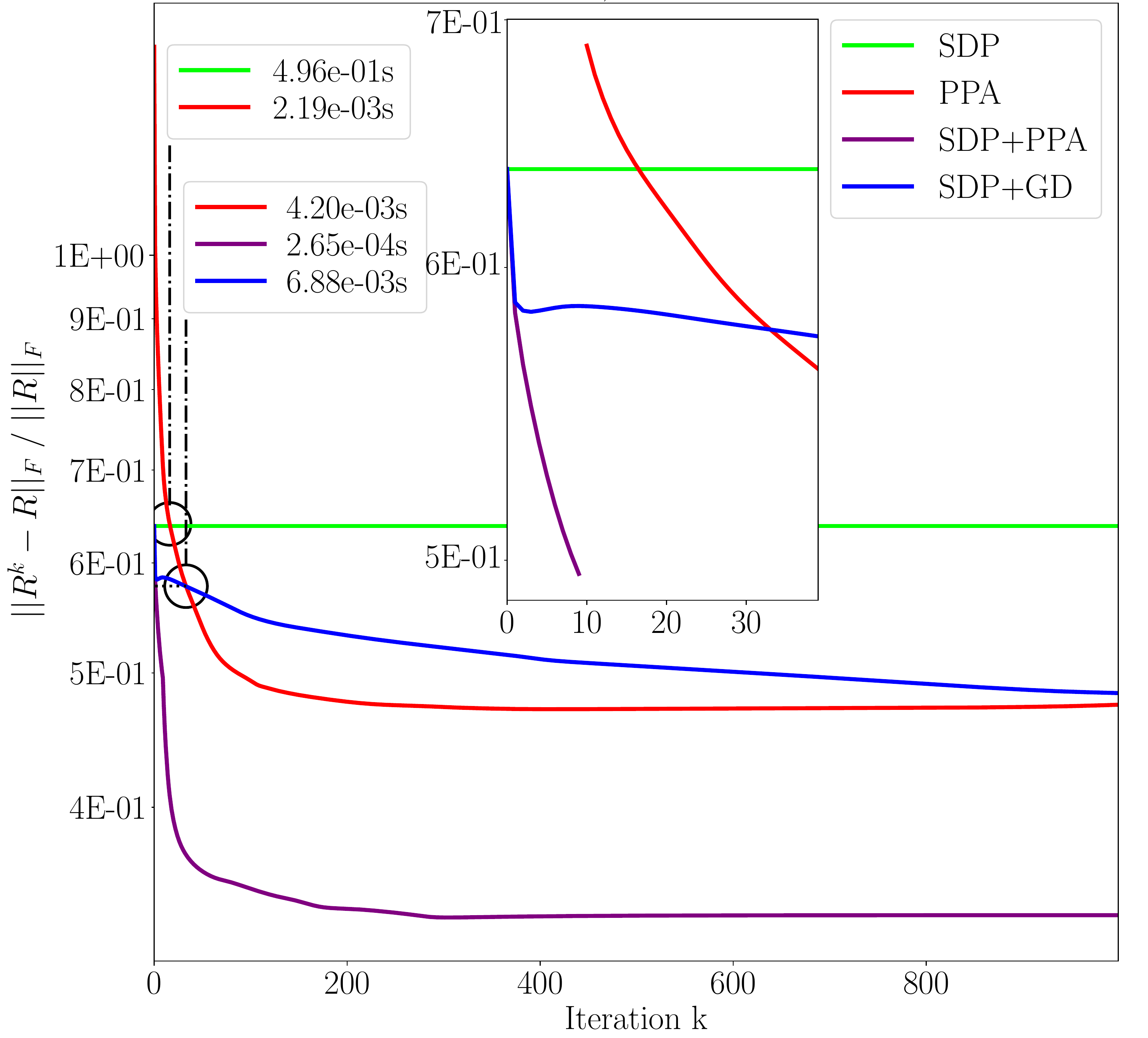}
		\label{ppa_sdp_a}}
	\subfigure[$\overline{t}=10$, $\text{SNR}=80$]{
		\includegraphics*[width=1.6in]{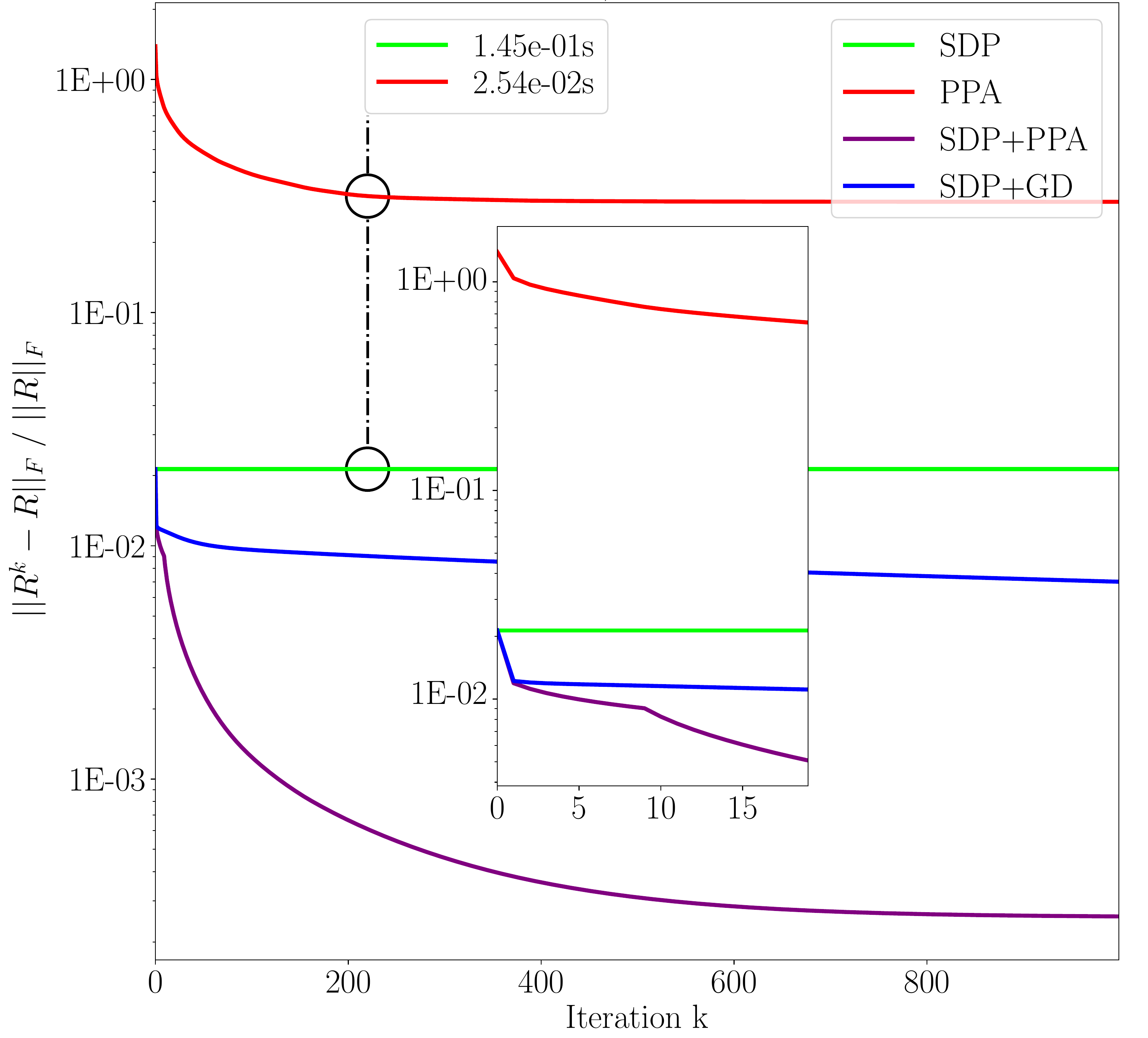}
		\label{ppa_sdp_b}}
	\subfigure[$\overline{t}=20$, $\text{SNR}=20$]{ % https://tug.org/TUGboat/tb34-1/tb106thurnherr.pdf
		\includegraphics*[width=1.6in]{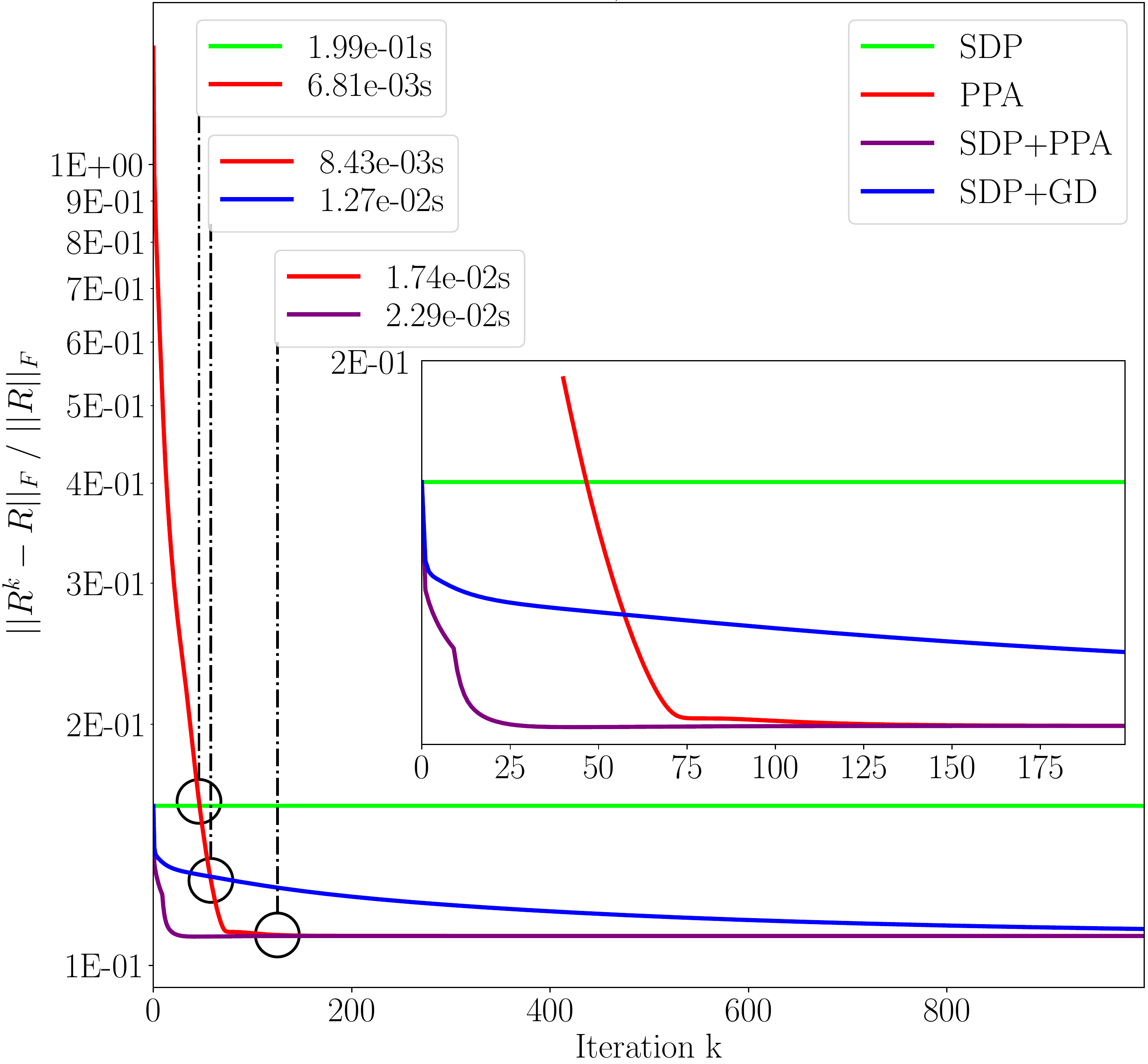}
		\label{ppa_sdp_c}}
	\subfigure[$\overline{t}=20$, $\text{SNR}=80$]{
		\includegraphics*[width=1.6in]{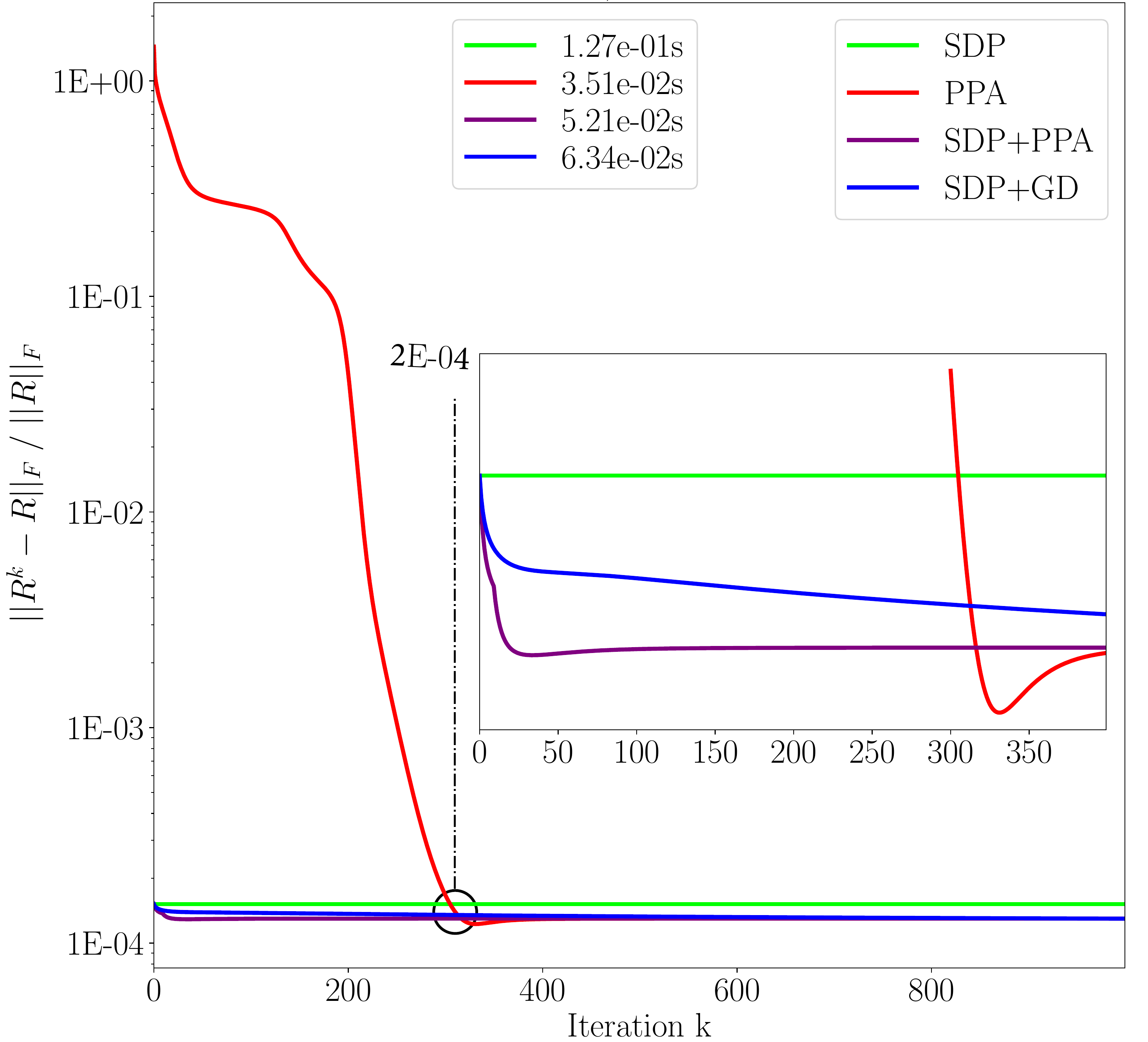}
		\label{ppa_sdp_d}}
	\subfigure[$\overline{t}=40$, $\text{SNR}=20$]{ % https://tug.org/TUGboat/tb34-1/tb106thurnherr.pdf
		\includegraphics*[width=1.6in]{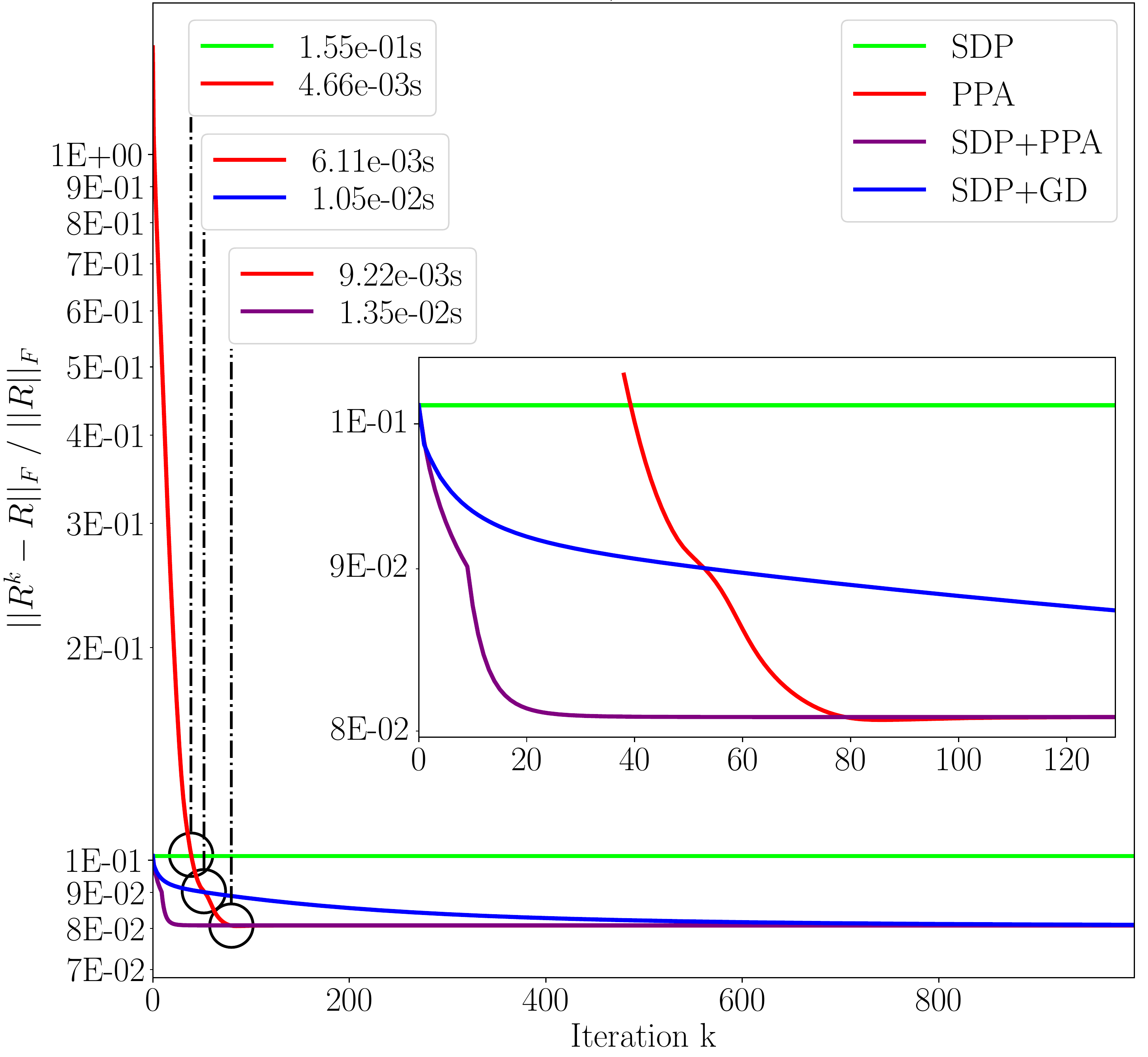}
		\label{ppa_sdp_e}}
	\subfigure[$\overline{t}=40$, $\text{SNR}=80$]{
		\includegraphics*[width=1.6in]{compare_num_20_snr_80}
		\label{ppa_sdp_f}}
	\caption{Convergence and running time of the PPA and the state-of-the-art algorithms for different $\overline{t}$ and SNR.}
	\label{fig:compare_ppa_sdp}
\end{figure}

\begin{table}[h!]
\centering
 \caption{Relative errors of translation (\%).} 
 \label{ret}
 \begin{tabular}{|c|c|c|c|c|}
 \hline
 Algorithms & SDP & PPA & SDP+PPA & SDP+GD\\
  \hline
 SNR=20 & 7.85 &  5.87  &  5.87 &5.99 \\
 \hline 
 SNR=30 & 6.20&4.48&4.48&5.12\\
 \hline SNR=80 & 0.09 &0.06&0.06&0.06\\ \hline
 \end{tabular}
\end{table}

Next, the performance of the RPA of Algorithm \ref{alg:opa2} is shown in Fig. \ref{rpa}, which illustrates that the relative error of the rotation matrix essentially decreases with the number of range measurements. Due to the use of approximation in deriving the RPA of \eqref{projonline}, it further induces performance degradation in comparison with the PPA. 
Note that the method in \cite{BOMIN17} is unable to write in a recursive form.

\begin{figure}[h!]
	\centering
			\includegraphics*[width=2.5in]{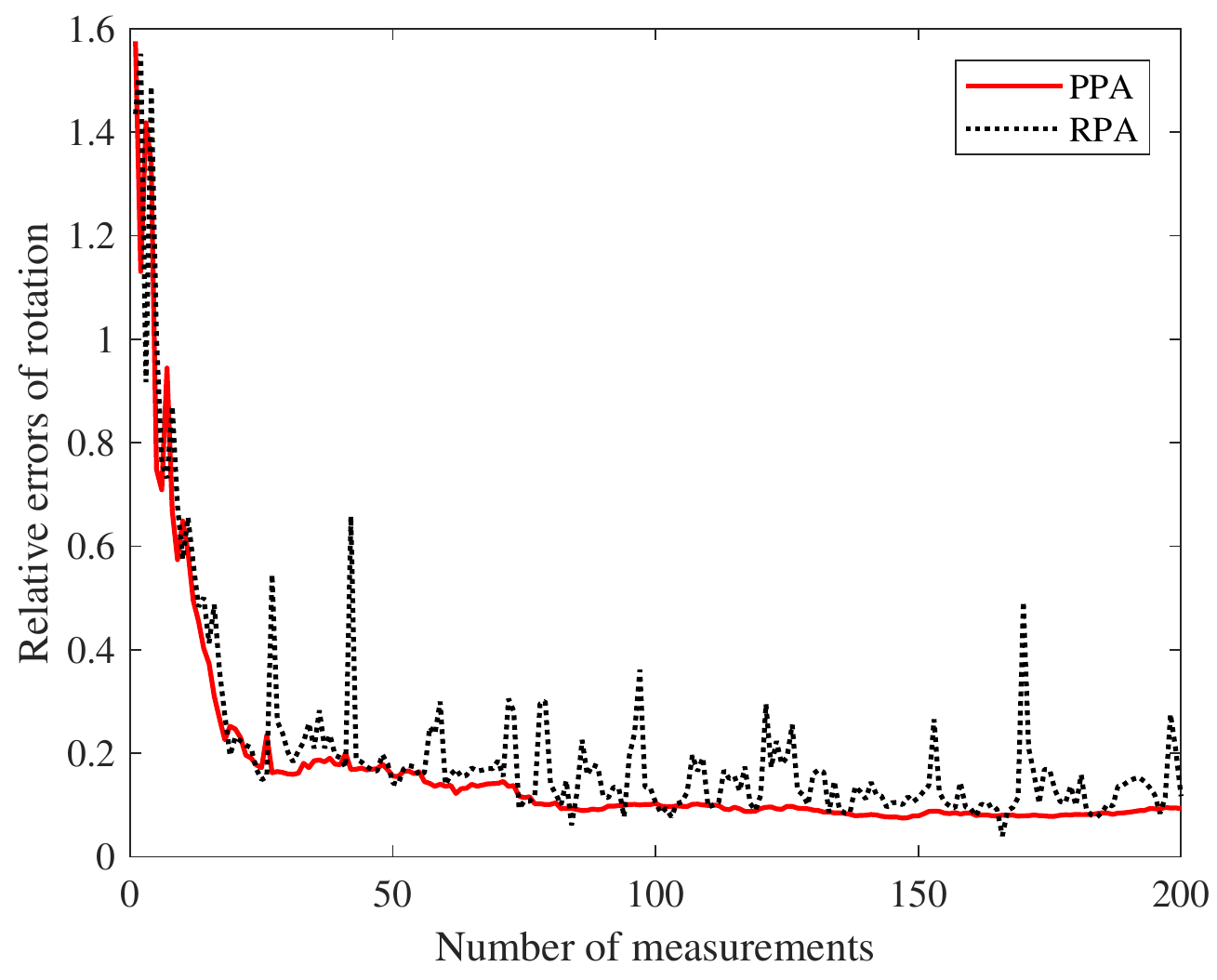}
		\label{ppa_sdp_a}
	\caption{Relative errors of rotation matrices by PPA and RPA for SNR=20.}
	\label{rpa}
\end{figure}

\subsection{Experimental results of the multi-node localization problem}
%We first consider time-varying networks consisting of three anchor nodes and three target nodes. To align the coordinate systems of target nodes, we use Algorithm \ref{alg:ppa3} under different initiations and the resulting estimates are presented in Fig.\ref{fig:position_pred_demo} and Fig.\ref{fig:position_pred_wrong_demo}, respectively where gray disks denote the estimated positions of target nodes. We observe that both cases well localize the target nodes.
%
%\begin{figure}[ht]
%	\centering
%	\subfigure[Network topology at $t=1$]{ % https://tug.org/TUGboat/tb34-1/tb106thurnherr.pdf
%		\includegraphics*[width=1.6in,viewport=140 130 680 640]{network_a_obj.pdf}
%		\label{fig:subfigure_network_a_pred}}
%	%\quad
%	\subfigure[Network topology at $t=2$]{
%		\includegraphics*[width=1.6in,viewport=140 130 680 640]{network_b_obj.pdf}
%		\label{fig:subfigure_network_b_pred}}
%	\caption{Localization results by using Algorithm \ref{alg:ppa3}. }
%	\label{fig:position_pred_demo}
%\end{figure}
%\begin{figure}[ht]
%	\centering
%	\subfigure[Network topology at $t=1$]{ % https://tug.org/TUGboat/tb34-1/tb106thurnherr.pdf
%		\includegraphics*[width=1.6in,viewport=140 130 680 740]{network_a_err.pdf}
%		\label{fig:subfigure_network_a_pred_wrong}}
%	%\quad
%	\subfigure[Network topology at $t=2$]{
%		\includegraphics*[width=1.6in,viewport=140 130 680 740]{network_b_err.pdf}
%		\label{fig:subfigure_network_b_pred_wrong}}
%	\caption{Localization results by using Algorithm \ref{alg:ppa3}.}
%	\label{fig:position_pred_wrong_demo}
%\end{figure}

In this subsection, we apply Algorithm \ref{alg:ppa3} to a sensor network which contains $110$ target nodes and $4$ anchor nodes. All the target nodes are randomly deployed in a two dimensional area $[1,9]\times[1,9]$, and $4$ anchor nodes are located at $(2,2),(2,8),(8,2),(8,8)$ respectively. Each node moves randomly in a unit square centered at its initial position. 

Two nodes can communicate only if their distance is within $1$, which clearly results in time-varying communication graphs. The SNR of each node is set to $\text{SNR}=100$. The localization results of target nodes at time slots $\overline{t}=5$ and $\overline{t}=25$ are presented in Fig. \ref{fig:cppa_results_test_loc}. We observe that the localization accuracy is improved when $\overline{t}$ increases and all target nodes are well localized.  

At the time slot $\overline{t}=25$, we count the number of  target-anchor range measurements for each target node, which is shown in Table \ref{fig:cppa_results_test_commu_25}. One can observe that in our cooperative localization method, more than a half ($50.9\%$) of target nodes have never directly taken range measurements with respect to any anchor node. However, their positions can also be successfully localized by Algorithm \ref{alg:ppa3} as shown in Fig. \ref{fig:cppa_results_test_loc_25}, which confirms the benefit of using {\em cooperative} methods.

\begin{figure}
	\centering
	\subfigure[Position estimation of target nodes at $\overline{t}=5$]{ % https://tug.org/TUGboat/tb34-1/tb106thurnherr.pdf
		\includegraphics*[width=2.3in]{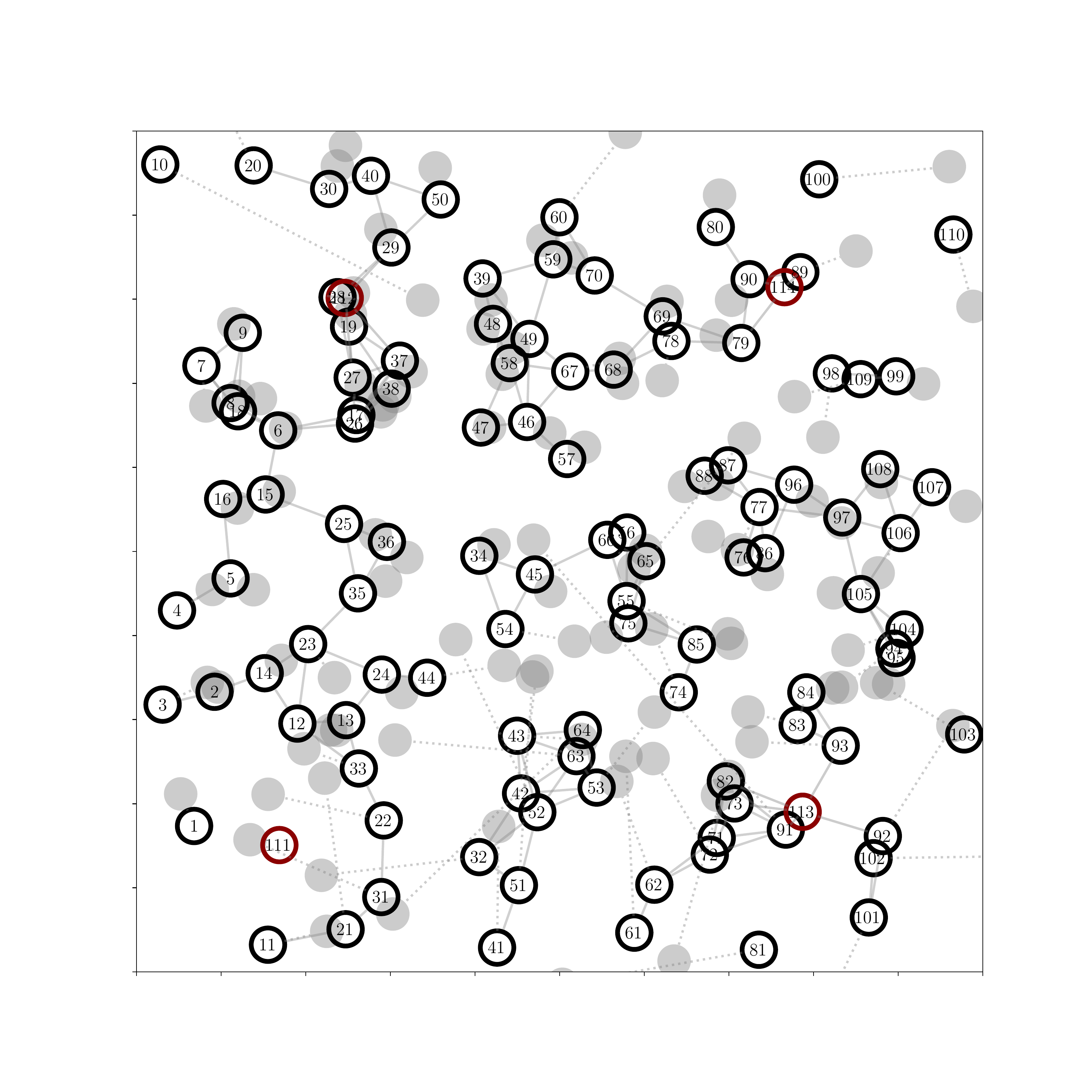}
		\label{fig:cppa_results_test_loc_5}}
%	\subfigure[Position estimation of target nodes at $\overline{t}=15$]{ % https://tug.org/TUGboat/tb34-1/tb106thurnherr.pdf
%		\includegraphics*[width=2.3in]{cppa_results_test_loc_15.pdf}
%		\label{fig:cppa_results_test_loc_15}}
	\subfigure[Position estimation of target nodes at $\overline{t}=25$]{ % https://tug.org/TUGboat/tb34-1/tb106thurnherr.pdf
		\includegraphics*[width=2.3in]{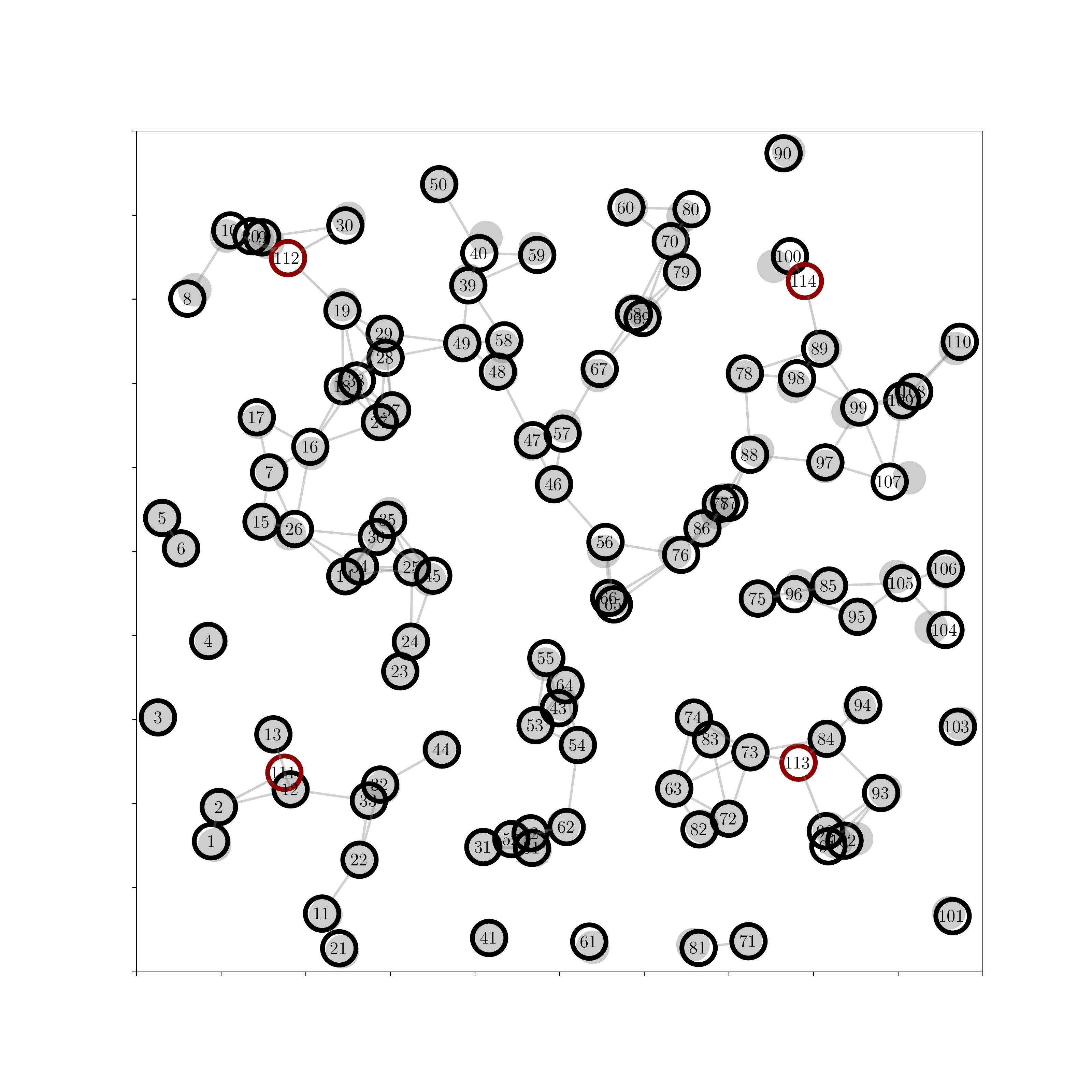}
		\label{fig:cppa_results_test_loc_25}}
	\caption{Localization results by Algorithm \ref{alg:ppa3}. {Red and black circles denote true positions of anchor nodes and target nodes. Gray circles are the estimated positions of target nodes. Gray edges represent the association between true positions and their estimated positions of target nodes.}}
	\label{fig:cppa_results_test_loc}
\end{figure}

%\begin{figure}[ht]
%	\begin{center}
%		\centering
%		\includegraphics*[width=3.5in]{cppa_results_test_commu_25}
%	\end{center}
%	\caption{Number of range measurements. The red bar represents the number of target-anchor range measurements for each target node, while the black bar represents the number of target-target range measurements.}
%	\label{fig:cppa_results_test_commu_25}
%\end{figure}

\begin{table}[t!] 
 \centering
 \caption{Distribution of target-anchor measurements over target nodes} 
 \label{fig:cppa_results_test_commu_25}
 \begin{tabular}{| >{\centering\arraybackslash}m{2.35cm} |cccccc|}
  \hline
  \# of target-anchor range measurements & $0$ & $\le 5$ & $\leq 10$ & $\leq 15$ & $\leq 20$ & $>20$\\ 
  \hline 
 \#  of target nodes & 56 &23 & 17&10&3&1 \\ 
    \hline 
  percentages (\%) & 50.9 & 20.9 & 15.5 & 9.1 & 2.7 & 0.9\\
  \hline
 \end{tabular}
\end{table}

Finally, we compare  Algorithm \ref{alg:ppa3} with the DPPA of Algorithm \ref{alg:ppa4} in a fixed graph $\mathcal{G}=(\mathcal{V},\mathcal{E})$ with {$110$ target nodes and $4$ anchor nodes}. Note that the DPPA is only applicable to a fixed graph. Define the average degree by
$$
\text{Deg} = {|\mathcal{E}|}/{|\mathcal{V}|}
$$
which characterizes the edge density of a network. By varying the average degree and the SNR, we implement both algorithms using the range measurements in a period of time $\overline{t}$ ($\overline{t}$ is set from $5$ to $25$). The resulting coordinate alignment relative errors are presented in Fig.\ref{fig:compare_dppa}, which illustrate that their performances are very close, and increasing the length of time interval $\overline{t}$ or the network density, both algorithms lead to better estimates. However, we recommend to use DPPA for a fixed graph as it involves simpler iterations and is a distributed version.

\begin{figure}[h!]
	\centering
	\subfigure[Mean relative error for SNR=10]{ 
		\includegraphics*[width=1.62in]{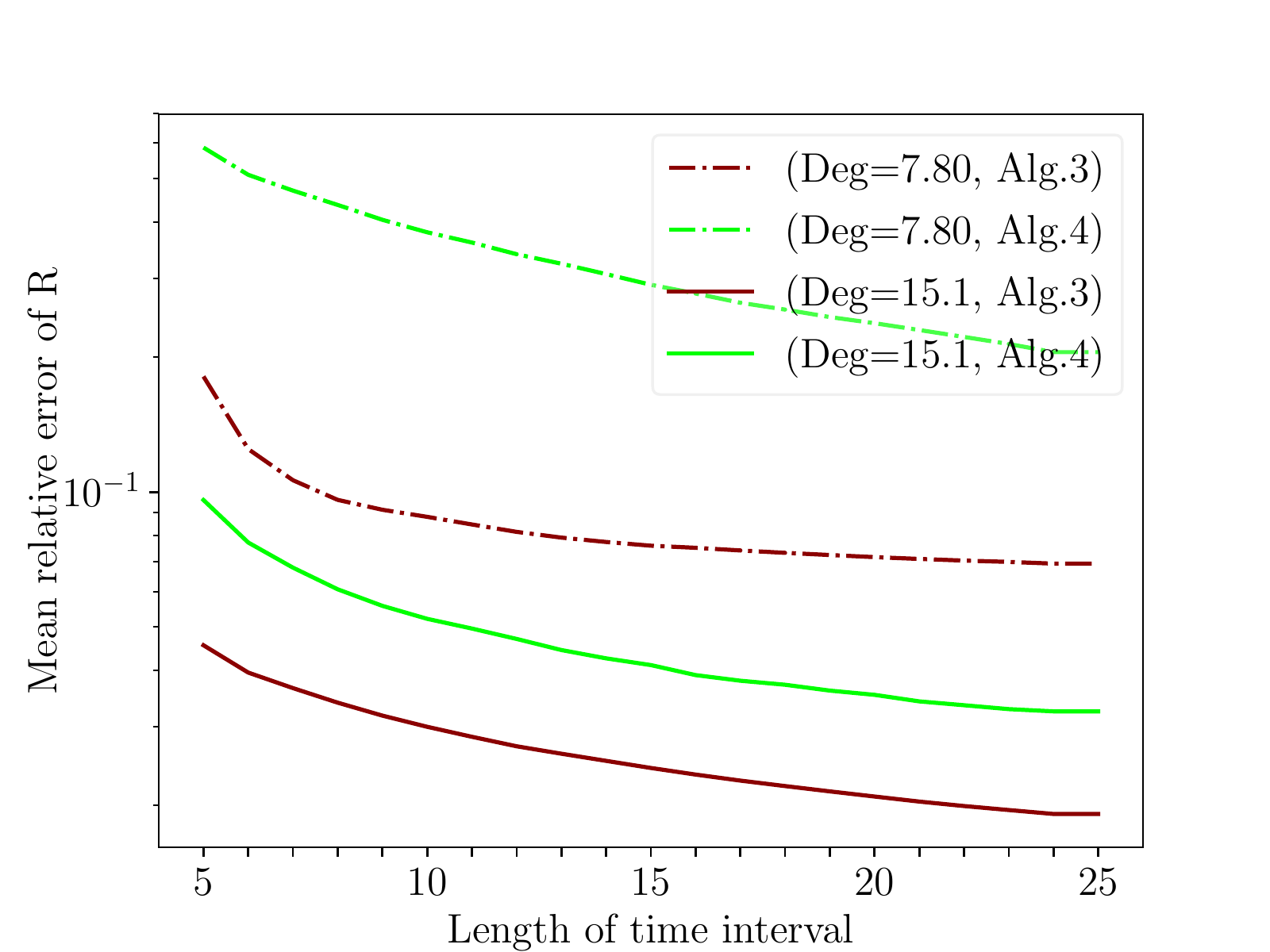}
		\label{fig:plot_compare_cppa_dppa_err_25}}
	\subfigure[Mean relative error for SNR=15]{
		\includegraphics*[width=1.62in]{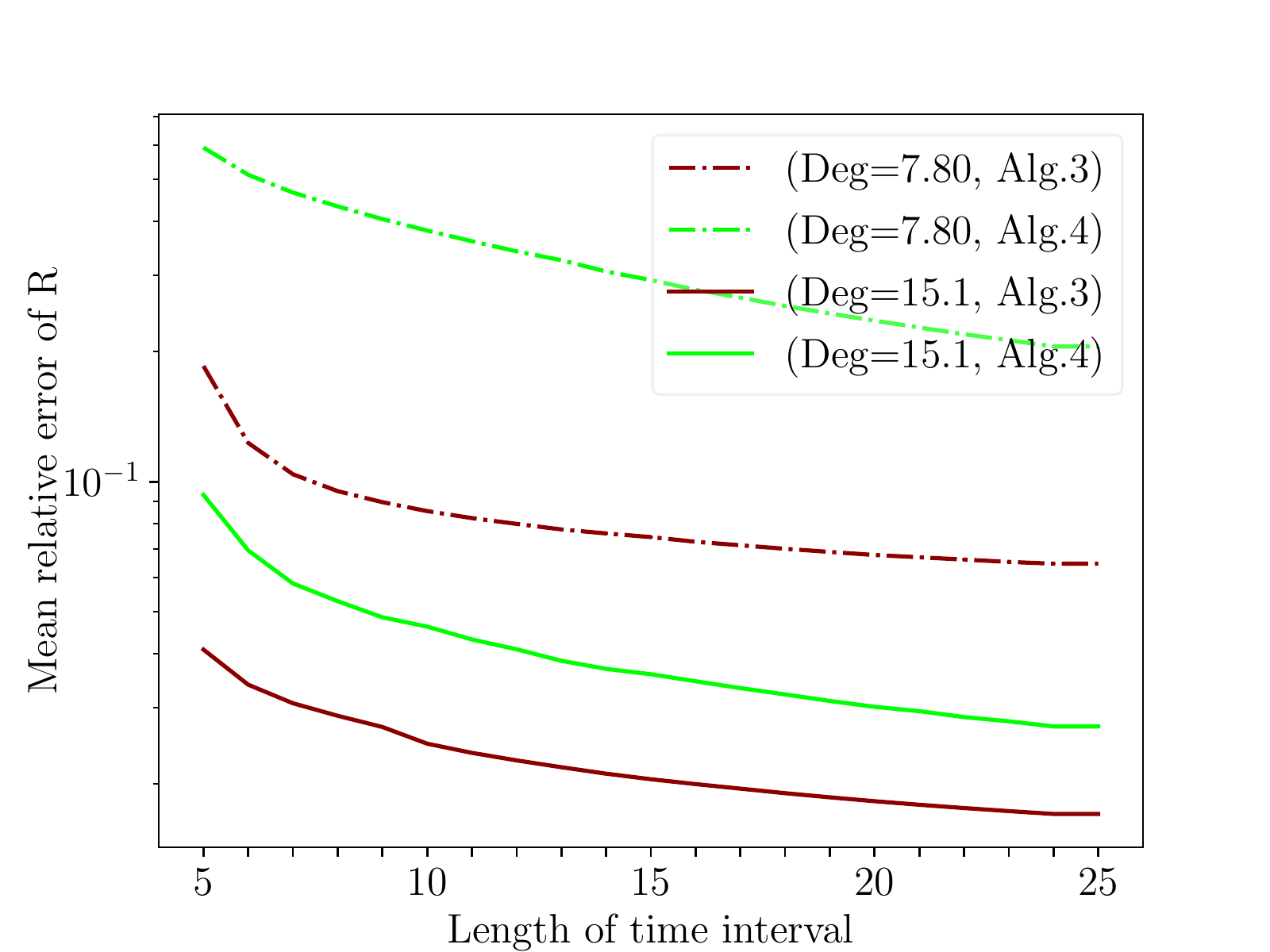}
		\label{fig:plot_compare_cppa_dppa_err_100}}
		\subfigure[Mean relative error for SNR=20]{
		\includegraphics*[width=1.62in]{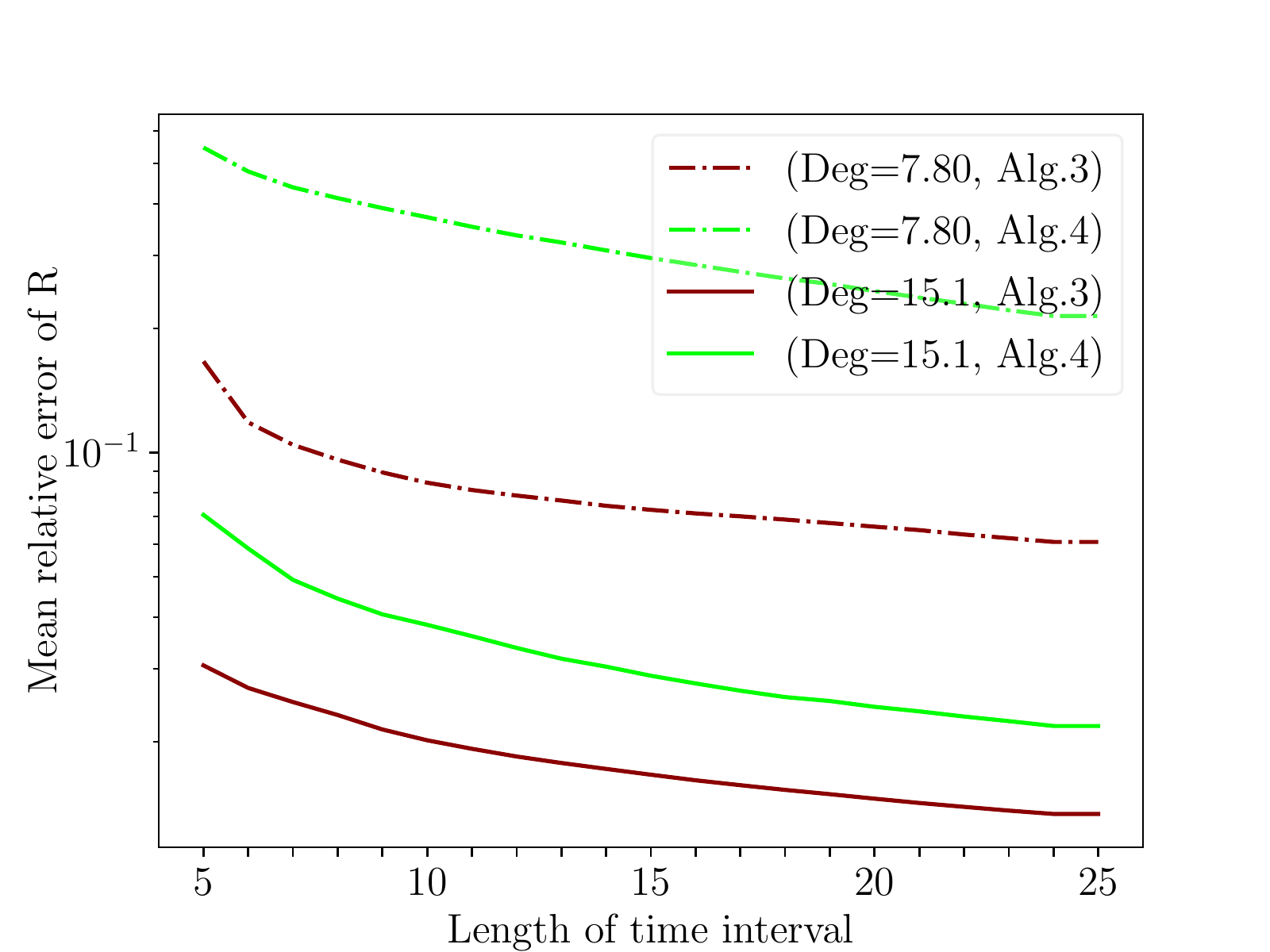}
		\label{fig:plot_compare_cppa_dppa_err_1000}}
	\subfigure[Mean relative error for SNR=35]{
		\includegraphics*[width=1.62in]{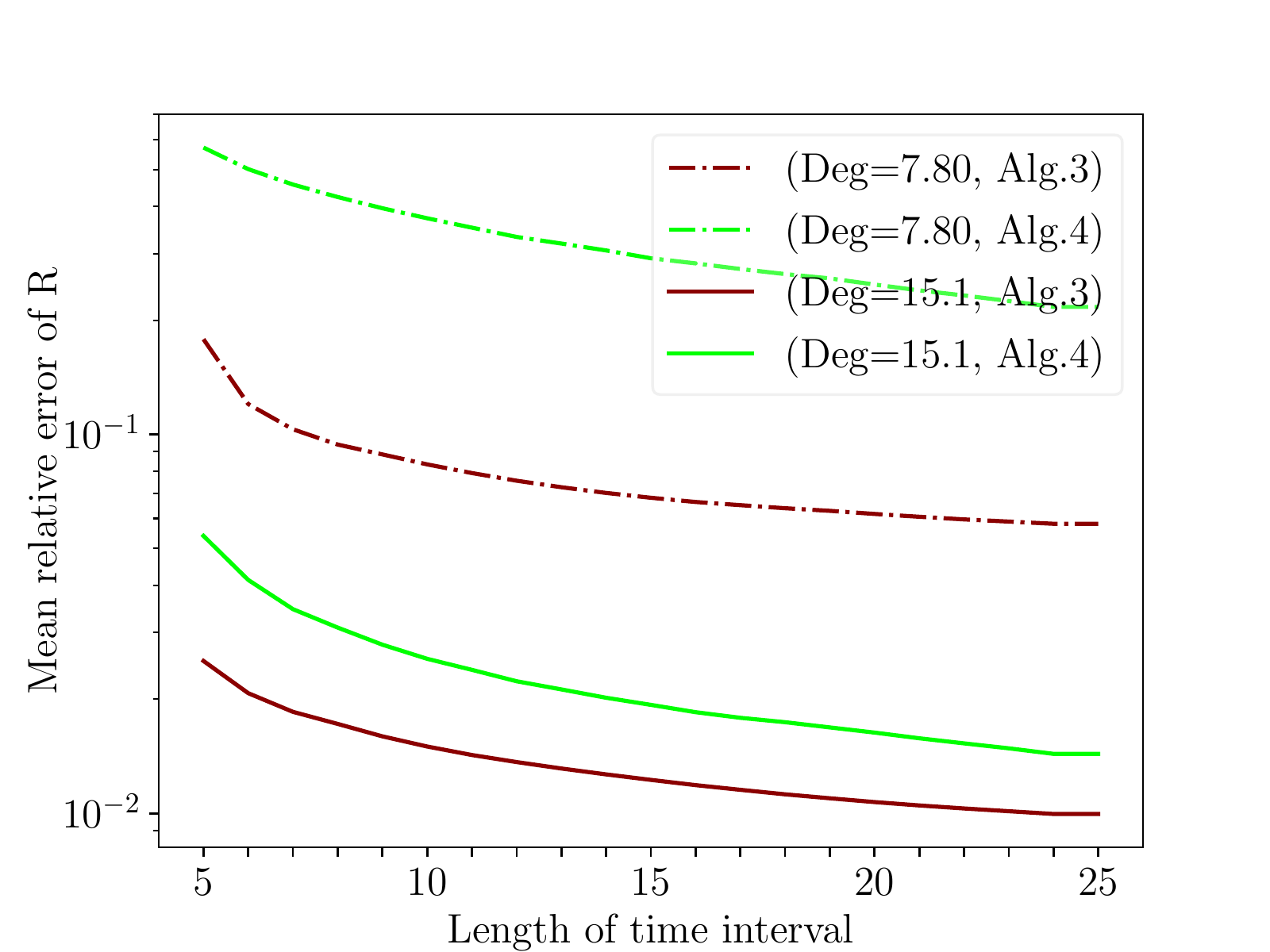}
		\label{fig:plot_compare_cppa_dppa_err_10000}}
	\caption{Relative errors by using Algorithm \ref{alg:ppa3} and  Algorithm \ref{alg:ppa4}.}
	\label{fig:compare_dppa}
\end{figure}

\section{Conclusion}\label{Conclusion}
This work considers the cooperative localization as a coordinate alignment problem using range measurements. To align the coordinate of a target node with an anchor node, we present PPA and RPA respectively. Then, the algorithms are generalized to the case of multiple target nodes in a sensor network. The effectiveness of all algorithms have been validated by numerical experiments. The state-of-the-art works such as the SDP and the SDP+GD are also compared with our work, which confirms the advantages of the proposed algorithms.

\bibliographystyle{IEEEtrans}
\bibliography{RefDBase}

\begin{thebibliography}{10}
\providecommand{\url}[1]{#1}
\csname url@samestyle\endcsname
\providecommand{\newblock}{\relax}
\providecommand{\bibinfo}[2]{#2}
\providecommand{\BIBentrySTDinterwordspacing}{\spaceskip=0pt\relax}
\providecommand{\BIBentryALTinterwordstretchfactor}{4}
\providecommand{\BIBentryALTinterwordspacing}{\spaceskip=\fontdimen2\font plus
\BIBentryALTinterwordstretchfactor\fontdimen3\font minus
  \fontdimen4\font\relax}
\providecommand{\BIBforeignlanguage}[2]{{%
\expandafter\ifx\csname l@#1\endcsname\relax
\typeout{** WARNING: IEEEtran.bst: No hyphenation pattern has been}%
\typeout{** loaded for the language `#1'. Using the pattern for}%
\typeout{** the default language instead.}%
\else
\language=\csname l@#1\endcsname
\fi
#2}}
\providecommand{\BIBdecl}{\relax}
\BIBdecl

\bibitem{patwari2005locating}
N.~Patwari, J.~N. Ash, S.~Kyperountas, A.~O. Hero, R.~L. Moses, and N.~S.
  Correal, ``Locating the nodes: cooperative localization in wireless sensor
  networks,'' \emph{IEEE Signal Processing Magazine}, vol.~22, no.~4, pp.
  54--69, 2005.

\bibitem{wymeersch2009cooperative}
H.~Wymeersch, J.~Lien, and M.~Z. Win, ``Cooperative localization in wireless
  networks,'' \emph{Proceedings of the IEEE}, vol.~97, no.~2, pp. 427--450,
  2009.

\bibitem{kia2015cooperative}
S.~S. Kia, S.~Rounds, and S.~Martinez, ``Cooperative localization for mobile
  agents: a recursive decentralized algorithm based on kalman-filter
  decoupling,'' \emph{IEEE Control Systems}, vol.~36, no.~2, pp. 86--101, 2016.

\bibitem{buehrer2018collaborative}
R.~M. Buehrer, H.~Wymeersch, and R.~M. Vaghefi, ``Collaborative sensor network
  localization: Algorithms and practical issues,'' \emph{Proceedings of the
  IEEE}, vol. 106, no.~6, pp. 1089--1114, 2018.

\bibitem{wang2006further}
Z.~Wang, S.~Zheng, S.~Boyd, and Y.~Ye, ``Further relaxations of the {SDP}
  approach to sensor network localization,'' Stanford University, Tech. Rep,
  Tech. Rep., 2006.

\bibitem{tseng2007second}
P.~Tseng, ``Second-order cone programming relaxation of sensor network
  localization,'' \emph{SIAM Journal on Optimization}, vol.~18, no.~1, pp.
  156--185, 2007.

\bibitem{nie2009sum}
J.~Nie, ``Sum of squares method for sensor network localization,''
  \emph{Computational Optimization and Applications}, vol.~43, no.~2, pp.
  151--179, 2009.

\bibitem{shang2004localization}
Y.~Shang, W.~Rumi, Y.~Zhang, and M.~Fromherz, ``Localization from connectivity
  in sensor networks,'' \emph{IEEE Transactions on Parallel and Distributed
  Systems}, vol.~15, no.~11, pp. 961--974, 2004.

\bibitem{soares2015simple}
C.~Soares, J.~Xavier, and J.~Gomes, ``Simple and fast convex relaxation method
  for cooperative localization in sensor networks using range measurements,''
  \emph{IEEE Transactions on Signal Processing}, vol.~63, no.~17, pp.
  4532--4543, 2015.

\bibitem{gholami2013cooperative}
M.~R. Gholami, L.~Tetruashvili, E.~G. Str{\"o}m, and Y.~Censor, ``Cooperative
  wireless sensor network positioning via implicit convex feasibility,''
  \emph{IEEE Transactions on Signal Processing}, vol.~61, no.~23, pp.
  5830--5840, 2013.

\bibitem{jia2011set}
T.~Jia and R.~M. Buehrer, ``A set-theoretic approach to collaborative position
  location for wireless networks,'' \emph{IEEE Transactions on Mobile
  Computing}, vol.~10, no.~9, pp. 1264--1275, 2011.

\bibitem{chen2017cooperative}
Q.~Chen, K.~You, and S.~Song, ``Cooperative localization for autonomous
  underwater vehicles using parallel projection,'' in \emph{13th IEEE
  International Conference on Control \& Automation}.\hskip 1em plus 0.5em
  minus 0.4em\relax IEEE, 2017, pp. 788--793.

\bibitem{BOMIN17}
B.~Jiang, B.~D. Anderson, and H.~Hmam, ``{3D} relative localization of mobile
  systems using distance-only measurements via semidefinite optimization,''
  \emph{Transactions on Aerospace and Electronic Systems, in press}, 2019.

\bibitem{bahr2009cooperative}
A.~Bahr, J.~J. Leonard, and M.~F. Fallon, ``Cooperative localization for
  autonomous underwater vehicles,'' \emph{The International Journal of Robotics
  Research}, vol.~28, no.~6, pp. 714--728, 2009.

\bibitem{papadopoulos2010cooperative}
G.~Papadopoulos, M.~F. Fallon, J.~J. Leonard, and N.~M. Patrikalakis,
  ``Cooperative localization of marine vehicles using nonlinear state
  estimation,'' in \emph{IEEE/RSJ International Conference on Intelligent
  Robots and Systems (IROS)}.\hskip 1em plus 0.5em minus 0.4em\relax IEEE,
  2010, pp. 4874--4879.

\bibitem{webster2012advances}
S.~E. Webster, R.~M. Eustice, H.~Singh, and L.~L. Whitcomb, ``Advances in
  single-beacon one-way-travel-time acoustic navigation for underwater
  vehicles,'' \emph{The International Journal of Robotics Research}, vol.~31,
  no.~8, pp. 935--950, 2012.

\bibitem{wang2014optimization}
S.~Wang, L.~Chen, D.~Gu, and H.~Hu, ``An optimization based moving horizon
  estimation with application to localization of autonomous underwater
  vehicles,'' \emph{Robotics and Autonomous Systems}, vol.~62, no.~10, pp.
  1581--1596, 2014.

\bibitem{allotta2016development}
B.~Allotta, A.~Caiti, R.~Costanzi, F.~Fanelli, E.~Meli, and A.~Ridolfi,
  ``Development and online validation of an ukf-based navigation algorithm for
  auvs,'' \emph{IFAC-PapersOnLine}, vol.~49, no.~15, pp. 69--74, 2016.

\bibitem{huang2018new}
Y.~Huang, Y.~Zhang, B.~Xu, Z.~Wu, and J.~A. Chambers, ``A new adaptive extended
  kalman filter for cooperative localization,'' \emph{IEEE Transactions on
  Aerospace and Electronic Systems}, vol.~54, no.~1, pp. 353--368, 2018.

\bibitem{bertsekas1999nonlinear}
D.~P. Bertsekas, \emph{Nonlinear Programming, 3rd edition}.\hskip 1em plus
  0.5em minus 0.4em\relax Athena Scientific, 2016.

\bibitem{yu2006principles}
C.~Yu, B.~Fidan, and B.~D. Anderson, ``Principles to control autonomous
  formation merging,'' in \emph{American Control Conference}.\hskip 1em plus
  0.5em minus 0.4em\relax IEEE, 2006, pp. 762--768.

\bibitem{naseri2017cooperative}
H.~Naseri and V.~Koivunen, ``Cooperative simultaneous localization and mapping
  by exploiting multipath propagation,'' \emph{IEEE Transactions on Signal
  Processing}, vol.~65, no.~1, pp. 200--211, 2017.

\bibitem{umeyama1991least}
S.~Umeyama, ``Least-squares estimation of transformation parameters between two
  point patterns,'' \emph{IEEE Transactions on Pattern Analysis and Machine
  Intelligence}, vol.~13, no.~4, pp. 376--380, 1991.

\bibitem{Andersson97}
L.~Andersson and T.~Elfving, ``A constrained procrustes problem,'' \emph{SIAM
  Journal on Matrix Analysis and Applications}, vol.~18, no. 124--139, 1997.

\bibitem{timmurphy}
``Procrustes analysis,''
  \url{https://www.mathworks.com/help/stats/procrustes.html}, accessed Feb 17,
  2020.

\bibitem{anderson2012optimal}
B.~D. Anderson and J.~B. Moore, \emph{Optimal Filtering}.\hskip 1em plus 0.5em
  minus 0.4em\relax Courier Corporation, 2012.

\bibitem{fong2011lsmr}
D.~C.-L. Fong and M.~Saunders, ``Lsmr: An iterative algorithm for sparse
  least-squares problems,'' \emph{SIAM Journal on Scientific Computing},
  vol.~33, no.~5, pp. 2950--2971, 2011.

\end{thebibliography}

\begin{IEEEbiography}
	[{\includegraphics[width=1in,height=1.25in,clip,keepaspectratio]{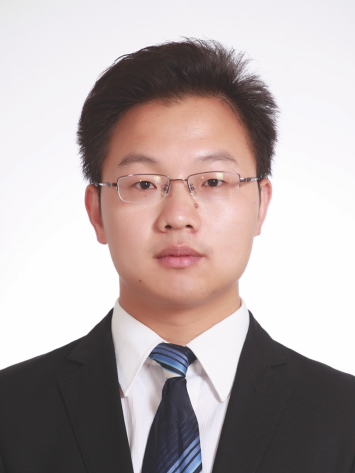}}]
	{Keyou You} (SM'17)  received the B.S. degree in Statistical Science from Sun Yat-sen University, Guangzhou, China, in 2007 and the Ph.D. degree in Electrical and Electronic Engineering from Nanyang Technological University (NTU), Singapore, in 2012.  Currently, he is a tenured Associate Professor in the Department of Automation, Tsinghua University, Beijing, China. His research interests include networked control systems, distributed optimization and learning, and their applications.
	Dr. You received the Guan Zhaozhi award in 2010 and the Asian Control Association Temasek Young Educator Award in 2019. He was selected to the National 1000-Youth Talent Program of China in 2014 and received the National Science Fund for Excellent Young Scholars in 2017. 
\end{IEEEbiography}
\vspace{-1cm}
\begin{IEEEbiography}
	[{\includegraphics[width=1in,height=1.25in,clip,keepaspectratio]{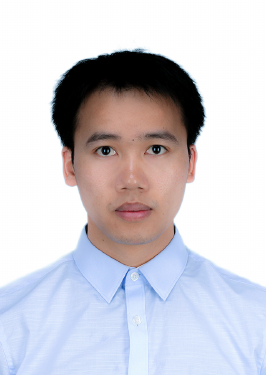}}]
{Qizhu Chen} received the B.E. degree in System Science and Engineering from Nanjing University, Nanjing, China, in 2015 and received the M.E degree in Control Science and Engineering from Tsinghua University, Beijing, China, in 2018. Currently, he is engaged in algorithm research and data mining in Beijing Science and Technology Co, three fast online. His research interests include distributed optimizations, machine learning, and their applications.
\end{IEEEbiography}
\vspace{-1cm}
\begin{IEEEbiography}
	[{\includegraphics[width=1in,height=1.25in,clip,keepaspectratio]{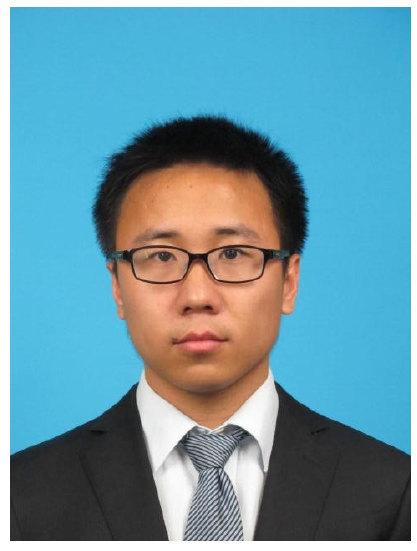}}]
{Pei Xie} received the B.E. degree and Ph.D degree in Control Science and Engineering from Tsinghua University, Beijing, China, in 2013 and 2019 respectively. Currently, he is engaged in algorithm research and data mining in JD.COM. His research interests include distributed optimizations, machine learning, operation research, and their applications.
\end{IEEEbiography}
\vspace{-1cm}
\begin{IEEEbiography}[{\includegraphics[width=1in,height=1.25in,clip,keepaspectratio]{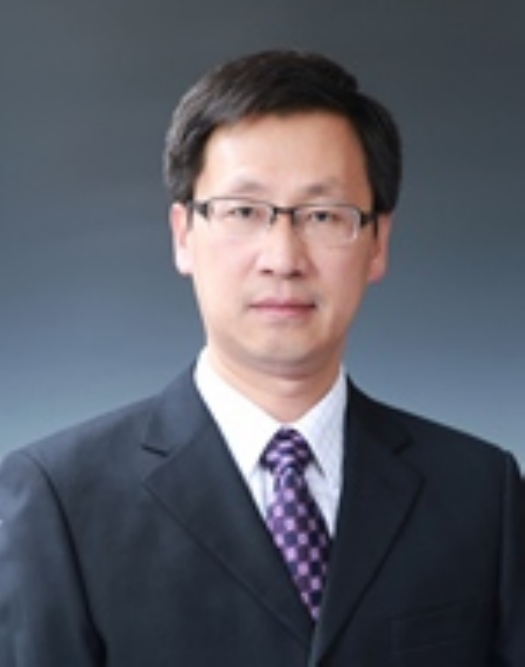}}]{Shiji Song}
received the Ph.D. degree in the Department of Mathematics from Harbin Institute of Technology in 1996. He is a professor in the Department of Automation, Tsinghua University. His research interests include system modeling, control and optimization, computational intelligence and pattern recognition.
\end{IEEEbiography}

\appendix
\subsection{Proof of Proposition \ref{prop_solvable}}\label{appendixa}
\begin{proof} We first prove that $f(R,T)$ is coercive \cite{bertsekas1999nonlinear} with respect to $T$, i.e., 
\bee\label{coercive}
\lim_{\twon{T}\rainfty}f(R, T)=\infty
\ene
where $\twon{T}=(\sum_{i=1}^n\twon{T_i}^2)^{1/2}$. Suppose that $\twon{T}\rainfty$, then there must exist some target node $i$ such that $\twon{T_i}\rainfty$. We have two exclusive scenarios.

 If $\bigcup_{t=1}^{\overline{t}}\cA_i(t)$ is nonempty,  e.g., there exists an anchor node $a$ such that $a\in\cA_i(t)$ for some $t$, i.e., the component $f_{ia}^{\cA}(t)$ exists in the objective function. Then, one can easily verify that $f_{ia}^{\cA}(t)$ tends to infinity as $\twon{T_i}\rainfty$. This implies that $\lim_{\twon{T_i}\rainfty}f(R,T)=\infty$. 

If $\bigcup_{t=1}^{\overline{t}}\cA_i(t)$ is empty, the target node $i$ must connect to an anchor node via some target node $j$ with a nonempty $\bigcup_{t=1}^{\overline{t}}\cA_j(t)$ in the union graph $\bigcup_{t=1}^{\overline{t}}\cG(t)$ since otherwise, the target node $i$ is disconnected to anchor nodes. Particularly, let $(j_0, j_1),\ldots, (j_{k-1},j_k)\in \bigcup_{t=1}^{\overline{t}}\cE(t)$ be the consecutive edges from node $i=j_0$ to node $j=j_k$. Suppose 
$$
\lim_{\twon{T_i}\rainfty}\sum_{t=1}^{\overline{t}}\sum_{v=0}^{k-1}f_{j_vj_{v+1}}^{\mathcal{T}}(t,R,T)<\infty,
$$
 it follows from (\ref{eq:sub_objective_function}) that $\twon{T_{j_0}}=\ldots=\twon{T_{j_k}}=\infty$. Since $\bigcup_{t=1}^{\overline{t}}\cA_j(t)$ is nonempty, it immediately implies that $\lim_{\twon{T_j}\rainfty}f(R,T)=\infty$. 

Overall, (\ref{coercive}) is proved.  Since the rotation group $\text{SO}(3)$ is compact and $f(R,T)$ is continuous, the rest of proof follows from the Weierstrass' theorem \cite[Proposition A.8]{bertsekas1999nonlinear}. 
\end{proof}

\subsection{Proof of Proposition \ref{prop_batch}}
\label{appendixb}
\begin{proof}
For any fixed $R \in \text{SO}(3)$, it is obvious that $T=\overline{y}^k-R\overline{p}^l$ minimizes the objective function of (\ref{master}) with respect to $T$. Let $T=\overline{y}^k-R\overline{p}^l$ in the objective function of (\ref{master}). Then, it follows that
\begin{equation*}
\begin{split}
&\sum\nolimits_{t=1}^{\overline{t}}{\lVert Rp^l(t) +(\overline{y}^k - R\overline{p}^l)- y^k(t)\rVert^2}\\ 
&= \sum\nolimits_{t=1}^{\overline{t}}{- 2\,(y^k(t) - \overline{y}^k)'R(p^l(t) - \overline{p}^l)} + c\\
&= -2\, \text{trace}\left( R'P^k \right) + c
\end{split}
\end{equation*}
where $c$ is independent of $R$ and is not explicitly given here. 

Then, $R^{k+1}$ is obtained via the minimization problem
\begin{equation*}
\begin{aligned}
R^{k+1}&=\argmin_{R\in\text{SO}(3)} -2\cdot \text{trace}\left( R'P^k \right)\\
&=\argmin_{R\in\text{SO}(3)}\twon{R-P^k}_F^2\\
&=P_{\text{SO}(3)}(P^k),
\end{aligned}
\end{equation*}
where the second equality follows from the fact that $\twon{R-P^k}_F^2=\text{trace}((R-P^k)'(R-P^k))=-2\text{trace}(R'P^k)+(P^k)'P^k+I$ for any $R\in \text{SO}(3)$.
\end{proof}

\subsection{Proof of Proposition \ref{prop_con}}
\label{appendixc}
{\begin{proof} By (\ref{eq:y_parallel_projection}) and $y^{k-1}(t)\in\cS(t), \forall t\in[1:\bar{t}]$, it follows
 $\twon{R^kp^l(t)+T^k-y^k(t)}\le \twon{R^kp^l(t)+T^k-y^{k-1}(t)},$
which implies that $$g(q^k)\le g(R^k, T^k,y_1^{k-1},\ldots,y_{\bar{t}}^{k-1}).$$ 
By  (\ref{master}), we obtain that
 $\sum\nolimits_{t=1}^{\overline{t}}{\lVert R^{k}p^l(t) + T^{k} - y^{k-1}(t) \rVert^2}\le \sum\nolimits_{t=1}^{\overline{t}}{\lVert Rp^l(t) + T - y^{k-1}(t) \rVert^2}$ for all $R\in\text{SO}(3)$ and $T\in\bR^3$. Since $R^{k-1}\in\text{SO}(3)$ and $T^{k-1}\in\bR^3$, this implies that 
 $$g(R^{k}, T^{k},y_1^{k-1},\ldots,y_{\bar{t}}^{k-1})\le g(q^{k-1}).$$
 Thus, it holds that $g(q^k)\le g(q^{k-1})$. 
Since $\text{SO}(3)$ and $\cS(t)$ are compact, it follows from (\ref{eq:r_t_solution}) that $\{q^k\}$ is a bounded sequence. Thus, it contains a convergent subsequence. 
\end{proof}}

\subsection{Proof of Proposition \ref{prof_rpa}}
\label{appendixd}
\begin{proof} Clearly, both $\bar{y}(t)$ and $\bar{p}^l(t)$ compute the time average of their associated vectors and can be expressed as
$$
\bar{y}(t)=\frac{1}{t}\sum_{i=1}^t y(i)~\text{and}~\bar{p}^l(t)=\frac{1}{t}\sum_{i=1}^t p^l(i).
$$
Moreover, it holds that
$$
P(t)=\sum_{i=1}^t (y(i)-\bar{y}(t))(p^l(i)-\bar{p}^l(t))'.
$$

In fact, let $\tilde{y}(t)=y(t)-\bar{y}(t-1)$ and $\tilde{p}^l(t)=p^l(t)-\bar{p}^l(t-1)$. Then, it follows from (\ref{ypl}) that
\bee
\begin{split}
P(t)&=\sum_{i=1}^t (y(i)-\bar{y}(t-1)-\frac{1}{t}\tilde{y}(t))\\
&~~~~~\times(p^l(i)-\bar{p}^l(t-1)-\frac{1}{t}\tilde{p}^l(t))'\\
&=P(t-1)+\frac{t-1}{t^2}\tilde{y}(t) \tilde{p}^l(t)' +(1-\frac{1}{t})^2\tilde{y}(t) \tilde{p}^l(t)' \\
&=P(t-1)+(1-\frac{1}{t})\tilde{y}(t) \tilde{p}^l(t)'.\nonumber
\end{split}
\ene

The rest of proof follows directly from that of Proposition \ref{prop_batch} and is omitted. 
\end{proof}

\end{document}